\documentclass[a4paper,11pt]{article}

\usepackage{libertine}
\usepackage[a4paper, margin=1in]{geometry}

\usepackage[backend=biber, style=numeric, doi=true, isbn=false, url=false, maxnames=6]{biblatex}
\AtBeginBibliography{}

\usepackage{pb-diagram}
\usepackage{amsmath, amssymb, amsthm}
\usepackage{tikz}
\usepackage{tikz-cd}
\usetikzlibrary{matrix,arrows,decorations.pathmorphing,arrows.meta,shapes.geometric, bending, positioning}

\usepackage[affil-it]{authblk}
\usepackage{hyperref}
\usepackage{paralist}
\usepackage{xcolor}
\usepackage{mathrsfs}

\usepackage{mathtools}
\usepackage[capitalise, noabbrev]{cleveref}
\usepackage{adjustbox}

\numberwithin{equation}{section}
\newtheorem{theorem}[equation]{Theorem}
\newtheorem{lemma}[equation]{Lemma}
\newtheorem{proposition}[equation]{Proposition}

\theoremstyle{definition}
\newtheorem{definition}[equation]{Definition}

\theoremstyle{remark}
\newtheorem*{remark}{Remark}

\makeatletter
\@namedef{subjclassname@2020}{%
  \textup{2020} Mathematics Subject Classification}
\makeatother

\DeclareMathOperator{\im}{im}
\DeclareMathOperator{\id}{id}

\DeclareMathOperator{\colim}{colim}

\DeclareMathOperator{\Ob}{Ob}

\def\Forget{\mathsf{O}}

\def\ProjFun{\mathsf{\Pi}}

\newcommand{\cell}{\mathsf{E}}

\def\Set{\mathbf{Set}}

\def\slice{{\downarrow}}
\def\dCat{\mathbf{dCat}}

\newcommand*\aut[1]{\mathbf{PAut}({#1})}

\newcommand*{\CSet}{\mathbf{cSet}}
\newcommand*{\cSet}{\mathbf{cSet}}

\def\cA{\mathcal{A}}
\def\cB{\mathcal{B}}
\def\cC{\mathcal{C}}
\def\cD{\mathcal{D}}
\def\cE{\mathcal{E}}

\def\cP{\mathcal{P}}
\def\cQ{\mathcal{Q}}

\def\cV{\mathcal{V}}

\def\slice{{\downarrow}}

\def\IP{\mathsf{IP}}
\def\src{\mathsf{src}}
\def\tgt{\mathsf{tgt}}
\def\intord{\dasharrow}

\def\ePos{\mathbf{ePos}}

\def\lcPos{\mathbf{lcPos}}

\def\yoneda{\mathsf{Y}}

\def\ecSet{\mathbf{ecSet}}
\def\lecSet{\mathbf{lecSet}}
\def\lecPos{\mathbf{lecPos}}

\def\lcSet{\mathbf{lcSet}}
\def\lcPos{\mathbf{lcPos}}

\def\Pos{\mathbf{Pos}}
\def\cPos{\mathbf{cPos}}
\def\ePos{\mathbf{ecPos}}
\def\ecPos{\mathbf{ecPos}}
\def\IPos{\mathbf{IPos}}

\def\cPos{\mathbf{cPos}}
\def\cIPos{\mathbf{cIPos}}
\def\ecIPos{\mathbf{ecIPos}}
\def\lecIPos{\mathbf{lecIPos}}

\def\Cosp{\mathbf{Cosp}}
\def\dCosp{\mathbf{dCosp}}
\def\dICosp{\mathbf{dICosp}}

\def\rng{\ast}

\def\Lang{\mathsf{Lang}}

\newcommand*\saut[1]{\mathbf{SAut}({#1})}
\newcommand*\tro[1]{\mathbf{TrO}({#1})}
\newcommand*\words[1]{\mathscr{W}({#1})}

\newcommand*\subsu{\sqsubseteq}
\newcommand*\psh[1]{\mathbf{PSh}({#1})}

\newcommand*\forrightarrow[1]{\xrightarrow[\adjustbox{raise={0.6ex}{0ex}}{$\scriptstyle{+}$}]{#1}}
\newcommand*\forleftarrow[1]{\xleftarrow[\adjustbox{raise={0.6ex}{0ex}}{$\scriptstyle{+}$}]{#1}}
\newcommand*\backrightarrow[1]{\xrightarrow[\adjustbox{raise={0.6ex}{0ex}}{$\scriptstyle{-}$}]{#1}}
\newcommand*\backleftarrow[1]{\xleftarrow[\adjustbox{raise={0.6ex}{0ex}}{$\scriptstyle{-}$}]{#1}}

\newcommand{\invLambda}{\rotatebox[origin=c]{180}{$\Lambda$}}

\begin{document}
\title{Track objects in cubical presheaf automata}
\date{}

\author[1]{Georg Struth\thanks{\url{g.struth@sheffield.ac.uk}}}
\author[2]{Krzysztof Ziemiański\thanks{\url{ziemians@mimuw.edu.pl}}}
\affil[1]{University of Sheffield, School of Computer Science}
\affil[2]{University of Warsaw, Institute of Mathematics}

\maketitle

\begin{abstract}
 Using the formalism of presheaf automata, we introduce new classes of automata for modelling concurrency that extend higher-dimensional automata and symmetric higher-dimensional automata. We study their languages and characterise these languages in terms of posetal cospans.
\end{abstract}

\noindent\textbf{MSC 2020:} Primary 68Q85, 18B20; Secondary
18B10, 68Q45.\\
\noindent\textbf{Keywords:} Higher-dimensional automaton, symmetric higher-dimensional automaton, track object, posetal cospan, cocompletion.




\section{Introduction}

Higher-dimensional automata (HDA) are a fundamental model of
non-interleaving concurrency. First proposed in the
1990s~\cite{Pratt_ModelingConcurrency}, their behaviour and relationship to other models
of concurrency, including Petri nets and event structures, have subsequently been explored~\cite{vanGlabbeek_Expressiveness}. A formal language theory for HDA, however, has only recently been developed~\cite{FJSZ_HDALang} and a Kleene theorem
has been established~\cite{FJSZ-Kleene} together with other analogues of classical theorems about automata \cite{FZ-Myhill_Nerode,ABFM_Logic}.
In parallel, a few related models of automata were studied, notably symmetric higher-dimensional automata \cite{Gaucher-Symmetric.Precubical.Sets,Kahl-SymmetricHDA,Zouari-Forgetting}.

The main goal of this article is to develop the basics of a language theory for variants of HDA, including symmetric HDA. Instead of following \cite{FJSZ_HDALang} directly, we use the formalism of presheaf automata \cite{SZ_Presheaf-Automata}, which allows a systematic categorical characterisation of languages for a wide class of  HDA-like automata in terms of posetal cospans.
We also propose new models for concurrent automata, in particular mixed HDA, which combine features of HDA and symmetric HDA. 

HDA are presheaves over labelled precube categories with distinguished start and accept elements (also called cells). Languages of HDA consist of ipomsets: isomorphism classes of finite partially ordered sets equipped with a labelling, interfaces and an event order \cite[Section 4]{FJSZ_HDALang}. One possible method of defining languages of HDA goes via track objects -- with each ipomset $P$ we associate its track object, an HDA $\square^P$, and the language of an HDA $X$ consists of all ipomsets that admit an HDA-map $\square^P\to X$.

The presheaf automata framework \cite{SZ_Presheaf-Automata} is a proposal for the categorification of automata theory. It builds on a notion of directed category, which is a category with two distinguished classes of morphisms: formorphisms and backmorphisms. The precube category is an example of a directed category: lower and upper coface maps are formorphisms and backmorphisms, respectively.
For a fixed directed category $\cC$, a ($\cC$-)presheaf automaton is a presheaf over $\cC$ with designated start and accept cells (elements).

Similarly to HDA, languages of presheaf automata over a directed category $\cC$ are defined in terms of track objects, though in this case the definition differs. A presheaf automaton that is simple (that is, has a single start cell and a single accept cell) can be interpreted as a cospan in the category $\psh\cC$ of presheaves over $\cC$: start and accept cells are replaced by maps from representable presheaves. Elementary track objects are images of directed cospans in $\cC$ (those having a formorphism and a backmorphism as their legs) under the Yoneda functor. Track objects are then defined as horizontal compositions of elementary ones. The language of a presheaf automaton $X$ is the class of track objects that admit a morphism into $X$.

The cost of generality in presheaf automata is that they provide no immediate description of track objects in terms of posets, or any other basic combinatorial objects. We provide such a description for a class of precube-like base categories. Namely, if $\cV$ is a subcategory of the directed category $\lecSet$ of labelled c-sets with event relation (Definition \ref{d:enriched-posets}) that is closed under subobjects (a sieve), then track objects over $\cV$ correspond to directed cospans in the category of interval $\cV$-posets whose feet are representable presheaves. Furthermore, this correspondence is a biequivalence of bicategories (Theorem \ref{t:cospans-are-track-objects}). These cospans can be regarded as interval pomsets with interfaces, which recovers the correspondence between track objects and interval ipomsets for classical HDA. Additionally, it gives an explicit description of languages of symmetric HDA and other related models.

We begin this article by recalling necessary facts about categories and presheaf automata in Section \ref{s:preliminaries}.
In Section~\ref{s:concsets}, we introduce a directed category $\cSet$ of c-sets (sets of concurrent events), which serves as the fundamental example of a directed category, as well as a base category for variants of precube categories; these are constructed as comma categories over $\cSet$.
This category is subsequently extended in Section~\ref{s:c-posets} to a directed category $\cPos$ of c-posets (partially ordered sets of events). We prove several properties of these categories; in particular, we show that directed pushouts exist in both $\cSet$ and $\cPos$, although the inclusion $\cSet\subset\cPos$ does not preserve them.

The central technical notion in this work is the directed cocompletion (d-cocompletion), introduced in Section \ref{s:d-cocompletions}. A directed category $\cP$ is a d-cocompletion of its full directed subcategory $\cC$ if it is closed under directed pushouts and generated from $\cC$ by directed pushouts (colimits of directed spans). Due to these properties, the Yoneda embedding $\cC\hookrightarrow\psh\cC$  extends fully and faithfully to $\cP$, and the resulting functor $\cP\to\psh\cC$ preserves directed pushouts.
Consequently, track objects over $\cC$ can be constructed as directed cospans in $\cP$ (rather than cospans in $\psh\cC$), as demonstrated in Section~\ref{s:d-cospans}. Furthermore, we identify so-called interval cospans in $\cP$ -- those corresponding to track objects (Theorem \ref{t:dcospans-are-track-objects}).

The primary example of a d-cocompletion is the inclusion $\cSet\subseteq\cPos$.
In Section~\ref{s:d-cocompletions}, we prove that, under certain conditions, d-cocompletions are preserved by the comma category construction (Theorem \ref{t:SlicedCocompletion}), which allows us to derive new examples of d-cocompletions from $\cSet\subseteq\cPos$. In Section~\ref{s:enriched-posets}, we equip c-sets and c-posets with a labelling and an event relation, which extends the event order of \cite{FJSZ_HDALang}. Then we show that the resulting categories of lec-sets $\lecSet$ and lec-posets $\lecPos$ are comma categories over $\cSet$ and $\cPos$, respectively, and hence the inclusion $\lecSet\subseteq\lecPos$ is a d-cocompletion.
Next, for every sieve $\cV\hookrightarrow\lecSet$ we define the subcategory of $\cV$-posets $\cV\Pos\subseteq\lecPos$ and show that the inclusion $\cV\subseteq\cV\Pos$ is a d-cocompletion as well. Particular choices of $\cV$ are variants of precube-like categories and these give rise to various examples of HDA-related automata.

In Section \ref{s:interval-posets}, we combine the results mentioned above and obtain a direct characterisation of track objects over a sieve $\cV\hookrightarrow\lecPos$. As a main result of the paper, we establish an equivalence between the bicategories of track objects over $\cV$ and of cospans of interval $\cV$-posets (Theorem \ref{t:cospans-are-track-objects}).

Finally, in Section~\ref{s:HDA}, we apply our results to several examples of sieves $\cV$, giving examples of HDA-related automata and descriptions of their respective languages. In this section, we adjust terminology where necessary to reduce the reliance on categorical abstraction and increase accessibility of the exposition.

\section{Preliminaries}
\label{s:preliminaries}

We recall some facts about categories from \cite{Adamek-Herrlich-Strecker-Joy_Of_Cats, Johnson-Yau_2Categories, MacLane_Categories,  Leinster:Basic_Bicategories} and more specifically about d-categories from~\cite{SZ_Presheaf-Automata},
and we fix notation.

\begin{paragraph}{Categories.}

All categories considered in this article are locally small. 
We write $\Ob(\cC)$ or simply $\cC$ for the class of objects of a category $\cC$ and $\cC(U,V)$ for the set of morphisms from $U$ to $V$.
A subcategory $\cD\hookrightarrow\cC$ is \emph{strictly full} if it is full and isomorphism-closed.  We uniformly write $\Forget:\cC\to \cD$ or $\Forget_{\cD}^{\cC}$ for the forgetful functor from category $\cC$ to category $\cD$. 
\end{paragraph}

\begin{paragraph}{Comma categories.}
Let $F\colon\cC\to\cD$ be a functor and  $L$ an object of $\cD$. Then $\cC\slice_F L$ denotes the category of \emph{objects $F$-over $L$}~\cite[II.6]{MacLane_Categories}. Its objects are pairs $(U,\lambda)$ for $U\in \Ob(\cC)$ and $\lambda\colon F(U)\to L$, a morphism $\alpha\colon (U,\lambda) \to (V,\mu)$ is a morphism $\alpha\colon U\to V$ satisfying $\lambda=\mu\circ F(\alpha)$.

We write $\ProjFun\colon\cC\slice_F L\to\cC$ for the projection functor.
If $F$ is the identity functor or clear from context, we write $\cC\slice L$ for $\cC\slice_F L$.
\end{paragraph}

\begin{paragraph}{Presheaves.}
A \emph{presheaf} on a category $\cC$ is a functor $X\colon \cC^{op}\to\Set$. We write $X[U]$ for the value of $X$ on $U\in\Ob(\cC)$ and  $\psh\cC$ for the category of presheaves on $\cC$. Further, $U^\wedge_\cC=\cC(-,U)$ denotes the representable presheaf on $U$, $f^\wedge_\cC\colon V^\wedge_\cC\to U^\wedge_\cC$ the morphism induced by $f\colon V\to U$ and $\yoneda_\cC\colon \cC\to \psh{\cC}$, $U\mapsto U^\wedge_\cC$ the Yoneda embedding.
    
Each functor $F\colon \cD\to\cC$ induces two functors between
categories of presheaves: the \emph{pullback} functor
$F^*\colon \psh\cC\to\psh\cD$, $X\mapsto X\circ F$, and the
\emph{pushforward} functor $F_*\colon \psh\cD\to\psh\cC$, its left
adjoint. The latter is the unique cocontinuous functor that fits into
the commutative square
    \[
    \begin{tikzcd}
        \cD\ar[d,"F"']\ar[r,"\yoneda_\cD"] 
        &
        \psh\cD\ar[d,"F_*"]
    \\
        \cC\ar[r,"\yoneda_\cC"]
        &
        \psh\cC
    \end{tikzcd}
    \]  

The \emph{category of elements} of a presheaf $X$ on $\cC$ is $\cell X=\cC\slice_{\yoneda_\cC} X$. Its objects are pairs $(U,x)$ of $U\in\Ob(\cC)$ and $x\in F(U)$ while
$\cell{X}((V,y),(X,x))=\{f\in \cC(V,U)\mid X[f](x)=y\}$ defines its morphisms.

The \emph{restricted Yoneda functor}~\cite{Isbell60} (see also~\cite{Ulmer68} and ~\cite[Chapters 3.6, 5.1]{Kelly82}) of a functor $J\colon \cC\to\cP$ is the composition
\[
\yoneda(J)\colon \cP\xrightarrow{\yoneda_\cP}\psh\cP\xrightarrow{J^*}\psh\cC,
\]
 where
$\yoneda(J)[U]=(J^\ast\circ \yoneda_\cP)[U] = \cP(J(-),U)$. In the special case of an inclusion $J\colon \cC\hookrightarrow \cP$ of a full subcategory we write
\begin{equation*}
    \yoneda_\cC^\cP=\yoneda(J)\colon \cP\to \psh\cC,
\end{equation*}
$P^\wedge_\cC$ for $\yoneda^\cP_\cC(P)$ and $f^\wedge_\cC\colon V^\wedge_\cC\to U^\wedge_\cC$ for the map induced by $f\in\cC(V,U)$. Note that $P^\wedge_\cC[U]=\cP(U,P)$.
\end{paragraph}

\begin{paragraph}{Directed categories.}
A category $\cC$ is \emph{directed} (a \emph{d-category}) if it is equipped with two wide subcategories $\cC^\alpha\subseteq\cC$, $\alpha\in\{{+},{-}\}$, such that each invertible morphism in $\cC$ is in $\cC^+$ if and only if its inverse is in $\cC^-$~\cite{SZ_Presheaf-Automata}. A \emph{d-functor} between d-categories $\cD$, $\cC$ is a functor $F\colon\cD\to \cC$ satisfying $F(\cD^\alpha)\subseteq \cC^\alpha$.

We call morphisms in $\cC^+$ \emph{forward morphisms} (\emph{formorphisms}) and those in $\cC^-$ \emph{backward morphisms} (\emph{backmorphisms}). We write $V\forrightarrow{}U$ and $V\backrightarrow{}U$ to indicate for- and backmorphisms. 

A d-functor $F\colon\cD\to \cC$ \emph{reflects} the d-structure if for every morphism $f$ in $\cD$, $f\in \cD^\alpha$ whenever $F(f)\in \cC^\alpha$.  
Each functor $F\colon\cD\to \cC$ into a d-category $\cC$ induces the \emph{reflected d-structure} on $\cD$ given by $\cD^\alpha = F^{-1}(\cC^\alpha)$.

In particular, if $\cC$ is a d-category, $F:\cC\to\cD$ a functor and $L\in\Ob(\cD)$, then the comma category $\cC\slice_FL$ is again a d-category with reflected d-structure induced by the projection functor \hbox{$\ProjFun:\cC\slice_FL\to\cC$}.

A d-category $\cD$ is a \emph{full d-subcategory} of a d-category $\cC$ if $\Ob(\cD)\subseteq\Ob(\cC)$ and $\cD(V,U)=\cC(V,U)$, $\cD^\alpha(V,U)=\cC^\alpha(V,U)$ for all $U,V\in\Ob(\cD)$. We then write  $\cD\hookrightarrow\cC$.

Every category can be equipped with the \emph{full d-structure} $\cC^\alpha = \cC$ and with the \emph{trivial d-structure} where only identities are for- and backmorphisms.
\end{paragraph}

\begin{paragraph}{Presheaf automata.}
A \emph{presheaf automaton} on a d-category $\cC$ (a \emph{$\cC$-automaton})~\cite[Def.\@ 3.6]{SZ_Presheaf-Automata} is a presheaf $X$ on $\cC$ with distinguished
sets of \emph{start elements} $\bot_X$ and
\emph{accept elements} $\top_X$ in $\cell X$. 
  
We write $\aut\cC$ for the category
of $\cC$-automata and morphisms of $\cC$-presheaves that preserve
start and accept elements:
\begin{equation*}
	\aut\cC(Y,X)=\{f\in\psh\cC(Y,X)\mid f(\bot_Y)\subseteq \bot_X,\; f(\top_Y)\subseteq \top_X\}.
\end{equation*}
\end{paragraph}

\begin{paragraph}{Simple presheaf automata.}

A $\cC$-automaton $X$ is \emph{simple}~\cite[Def.\@ 3.7]{SZ_Presheaf-Automata} if it has a single start
element $\bot_X$ and a single accept element $\top_X$ -- written with some abuse of notation. The  \emph{source} and \emph{target} of a simple $\cC$-automaton $X$ are $\src(X)=\pi_X(\bot_X)\in\Ob(\cC)$ and $\tgt(X)=\pi_X(\top_X)\in\Ob(\cC)$, respectively. 

We write $\saut{\cC}$ for the full subcategory of simple $\cC$-automata within $\aut\cC$. It decomposes into a disjoint union
\[
    \saut\cC\cong \coprod_{U,V\in \Ob(\cC)} \saut\cC_U^V,
\]
where $\saut\cC_U^V$  is the full subcategory  of $\saut\cC$ whose objects $X$ satisfy $\src(X)=U$ and $\tgt(X)=V$.
    
Each simple $\cC$-automaton $X\in\saut\cC$ can be regarded as a cospan 
\begin{equation*}
	\src(X)^\wedge_\cC\xrightarrow{\bot_X} X \xleftarrow{\top_X}
        \tgt(X)^\wedge_\cC
\end{equation*}
of $\cC$-presheaves, abusing notation for $\bot_X$ and $\top_X$ once again: they are now morphisms of presheaves.
A morphism $f:X \to Y$ of simple $\cC$-automata is then a commutative diagram ($U=\src(X)=\src(Y)$, $X=\tgt(X)=\tgt(Y)$)
\begin{equation}
\label{e:saut-morphisms}
    \begin{tikzcd}[row sep=tiny]
        & X\ar[dd,"f"] & \\
        U^\wedge_\cC\ar[ru,"\bot_X"]\ar[rd,"\bot_Y"'] & & V^\wedge_\cC\ar[lu,"\top_X"']\ar[ld,"\top_Y"] \\
        & Y & 
    \end{tikzcd}
\end{equation}

The \emph{concatenation} of simple $\cC$-automata $X\in\saut\cC_U^V$ and $Y\in\saut\cC_V^W$ is the simple $\cC$-automaton $X*Y\in\saut\cC_U^W$ constructed by the pushout
\begin{equation}
\label{e:saut-concatenation}
    \begin{tikzcd}[row sep=0.5cm, column sep=0.5cm]
        & & X*Y 
        & & \\
        & X\ar[ru] & & Y \ar[lu] & \\
        U^\wedge_\cC 
        \ar[ru,"\bot_X"']
        \ar[uurr, bend left=30, "\bot_{X*Y}"]
        & &
        V^\wedge_\cC \ar[ru,"\bot_Y"] \ar[lu,"\top_X" swap] & &
        W^\wedge_\cC 
        \ar[lu,"\top_Y"]
        \ar[uull, bend right=30, "\top_{X*Y}"']
    \end{tikzcd}
\end{equation}
This defines functors $\saut\cC_U^V\times \saut\cC_V^W\to \saut\cC_U^W$ for all $U,V,W\in\Ob(\cC)$. 

Concatenation of simple $\cC$-automata is associative up to natural isomorphism. Thus the category $\saut\cC$ gives rise to an example of a cospan bicategory, with concatenation of simple $\cC$-automata being horizontal cospan composition.  We discuss this connection further below.
\end{paragraph}

\begin{paragraph}{Track objects.}

An \emph{elementary track object} is a simple $\cC$-automaton 
\begin{equation*}
V^\wedge_\cC\xrightarrow{f^\wedge_\cC} U^\wedge_\cC \xleftarrow{g^\wedge_\cC} W^\wedge_\cC,
\end{equation*}
for some $U,V,W\in\Ob(\cC)$, $f\in\cC^+(V,U)$ and $g\in\cC^-(W,U)$.
A \emph{track object} over $\cC$ (\emph{$\cC$-track object}) is a simple automaton isomorphic to a finite concatenation of elementary track objects. 
   
The (strictly) full subcategory $\tro\cC$ of track objects of $\saut\cC$ decomposes again into a disjoint union 
\begin{equation*}
\tro\cC=\coprod_{U,V\in\Ob(\cC)}\tro\cC_U^V,
\end{equation*}
where $\tro\cC_U^V=\tro\cC\cap \saut\cC_U^V$. We call morphisms of $\tro\cC$ \emph{subsumptions}. 

The definition given here is equivalent to the original one \cite[Def.\@ 7.1]{SZ_Presheaf-Automata} by \cite[Lem.\@ 7.13]{SZ_Presheaf-Automata}.
\end{paragraph}

\begin{paragraph}{Languages.}
We slightly generalise a previous definition~\cite[Def.\@ 8.1]
{SZ_Presheaf-Automata}.
A \emph{$\cC$-word} in a d-category $\mathcal{C}$ is an equivalence class of track objects with respect to the relation $\sim$ defined, for all $\Gamma,\Delta\in\tro\cC$, by
\[
\Gamma\sim\Delta \iff \tro{\cC}(\Gamma,\Delta)\neq\emptyset\neq\tro{\cC}(\Delta,\Gamma).
\]
The class of $\cC$-words forms a poset $\words\cC$ with respect to the \emph{subsumption relation}
\[
    [\Gamma]\subsu[\Delta]
    \iff
    \tro\cC(\Gamma,\Delta)\neq \emptyset.
\]
The category of $\cC$-words has the same objects as $\cC$ and words $\mathscr{W}(\cC)_U^V=(\tro\cC_V^U)/{\sim}$ as morphisms $U\to V$. Word composition is induced by concatenation of track objects. Subsumption defines an enrichment of the category $\mathscr{W}(\cC)$ in the category of posets.

The \emph{language} of a $\cC$-automaton $X$~\cite[Def.\@ 8.2]{SZ_Presheaf-Automata} is the class of $\cC$-words
\[
    \Lang(X) = \{[\Gamma]\in \words\cC\mid \aut\cC(\Gamma,X)\neq \emptyset \}.
\]

The class $\Lang(X)$ is down-closed with respect to subsumption (if $[\Gamma]\in\Lang(X)$ and $[\Delta]\subsu[\Gamma]$, then $[\Delta]\in \Lang(X)$).
For specific choices of $\cC$ we recover languages of finite automata, higher-dimensional automata, Petri nets, etc -- see \cite{SZ_Presheaf-Automata} for examples and a further details on languages of $\cC$-automata.
\end{paragraph}

\begin{paragraph}{Directed pushouts and pullbacks.}
Consider the d-categories
\begin{equation*}    
    \Lambda=
    \left(
    \begin{tikzcd}
        0 \ar[r,"-" swap]\ar[d,"+"] & 1 \\
        2
    \end{tikzcd}
    \right)
    \qquad
    \lozenge=
    \left(
    \begin{tikzcd}
        0 \ar[r,"-" swap]\ar[d,"+"]\ar[dr] & 1\ar[d,"+" swap] \\
        2\ar[r,"-"] & 3
    \end{tikzcd}
    \right)
    \qquad
    \invLambda
    =
    \left(
    \begin{tikzcd}
         & 1\ar[d,"+" swap] \\
        2\ar[r,"-"] & 3
    \end{tikzcd}    
    \right)
\end{equation*}
(We occasionally  use these symbols also for the corresponding categories without d-structure.)

Let $\cC$ be a d-category.
\begin{itemize}
    \item A \emph{d-span} (\emph{d-cospan}) is a d-functor $\Lambda\to\cC$ ($\invLambda\to \cC$).
    \item A \emph{d-square} is a d-functor $\lozenge\to \cC$ (d-squares commute by definition).
    \item A \emph{d-pushout} (\emph{d-pullback}) is a d-square that is also a pushout (pullback).
\end{itemize}

A d-category $\cC$ is \emph{d-cocomplete} if all d-spans extend to d-pushouts (in which case the structural inclusions must be for- and backmorphisms, respectively).
 
A d-functor $F\colon \cD\to \cC$
\begin{itemize}
    \item \emph{preserves d-pushouts} if it maps d-pushouts to d-pushouts,
    \item \emph{reflects} d-pushouts if $X$ is a d-pushout whenever a d-square $\lozenge \xrightarrow{X} \cD \xrightarrow{F} \cC$ is,
    \item \emph{creates d-pushouts} if 
 every lifting problem 
\begin{equation}
\label{e:pushout-creation}
\begin{tikzcd}
\Lambda\ar[d]\ar[r,"Y"] & \cD\ar[d,"F"] \\
\lozenge\ar[r, swap, "X"]\ar[ru,dashed,"Y'"] & \cC
\end{tikzcd}
\end{equation}
for the d-pushout $X$ in $\dCat$ has a unique solution $Y'$ which is a d-pushout. 
\end{itemize}
A functor that creates d-pushouts also preserves and reflects them. Preservation, reflection and creation of d-pullbacks is defined similarly.
\end{paragraph}

\begin{paragraph}{Bicategories.}
A \emph{bicategory} $\cB$ consists of
\begin{itemize}
    \item a class of objects $\Ob(\cB)$, 
    \item a hom category $\cB_U^V$ for
    all $U,V\in \Ob(\cB)$,
    \item horizontal composition functors $\cB_U^V\times \cB_V^W\to \cB_U^W$ for all $U,V,W\in\Ob(\cB)$, 
    \item identities, associators and unitors subject to coherence conditions.
\end{itemize} 
See \cite[2.1]{Johnson-Yau_2Categories} for details. Morphism composition in hom categories is called \emph{vertical composition}.

A \emph{pseudofunctor} $F$ from a bicategory $\cA$ to a bicategory $\cB$ consists of a function $\Ob(F)\colon\Ob(\cA)\to\Ob(\cB)$, \emph{local functors} $F_U^V\colon\cA_U^V\to\cB_U^V$, and natural transformations called laxity constraints~\cite[4.1]{Johnson-Yau_2Categories}. 
A pseudofunctor is \emph{locally faithful} (\emph{locally full}) if all its local functors are faithful (full).

A pseudofunctor $F\colon\cA\to\cB$ is a \emph{biequivalence} \cite[2.2]{Leinster:Basic_Bicategories} \cite[6.2.8]{Johnson-Yau_2Categories} if it is locally an equivalence (that is, local functors between hom categories are equivalences) and surjective up to equivalence on objects. Clearly, the latter holds when $\Ob(F)$ is an identity.
    
More generally, for a given property $P$, we say that a bicategory is \emph{locally $P$} if all hom categories satisfy $P$. For example, a sub-bicategory is locally full if all local subcategories are full; a pseudofunctor is locally faithful if all local functors are faithful.
\end{paragraph}

\begin{paragraph}{Cospans.}
Let $\cC$ be a category that has all pushouts. 
The \emph{bicategory of $\cC$-cospans} \cite[2.1.22]{Johnson-Yau_2Categories} $\Cosp(\cC)$ is given by the following data:
\begin{itemize}
\item Its objects are those of $\cC$. 
\item The hom category $\Cosp(\cC)_V^W$ has cospans $\big(V\to U\gets W\big)$ in $\cC$ as objects ($U,V,W\in\Ob(\cC)$), and commutative diagrams as in \eqref{e:saut-morphisms} in $\cP$ as morphisms.
\item Horizontal composition of cospans is given by
\[
\big(V\xrightarrow{f} U \xleftarrow{g} W\big)
*
\big(W\xrightarrow{h} X \xleftarrow{i} Y\big)
=
\big(V \xrightarrow{j\circ f} U\sqcup_{W}X\xleftarrow{k\circ i} Y\big),
\]
where $j:U\to U\sqcup_{W}X$ and $k:X\to U\sqcup_{W}X$ are canonical inclusions, similarly to \eqref{e:saut-concatenation}.
\end{itemize}

For a class $\mathscr{A}\subseteq \Ob(\cC)$, we denote by $\Cosp(\cC,\mathscr{A})\subseteq\Cosp(\cC)$ the full sub-bicategory with objects $\mathscr{A}$. With some abuse of notation, we write
\begin{itemize}
\item 
    $\Cosp(\cC,\cD)$ for $\Cosp(\cC,\Ob(\cD))$ if $\cD$ is a subcategory of $\cC$,
\item 
    $\Cosp(\psh\cC,\cC)$ for the full sub-bicategory of $\Cosp(\psh\cC)$ that has presheaves represented by objects of $\cC$ as objects. 
\end{itemize}
See Definition~\ref{d:d-cospan} for a generalisation of this construction.

The bicategory $\Cosp(\psh\cC,\cC)$ has hom categories $\saut\cC_V^W$ for $V,W\in\Ob(\cC)$, and the horizontal composition is given by concatenation. The class of track objects over $\cC$ forms the locally strictly full sub-bicategory $\tro\cC\subseteq \Cosp(\psh\cC,\cC)$.
It is a minimal locally strictly full sub-bicategory that contains all elementary track objects.
\end{paragraph}

\section{Concurrent sets}
\label{s:concsets}

In this section we introduce the d-category of concurrent sets and study its basic properties. We show that this category is d-cocomplete and has all pullbacks.

\begin{definition}
A \emph{c-set} is a finite set.
\end{definition}

Here, ``c-set'' stands for ``concurrent set'' because c-sets represent sets of concurrent events in the intended applications of higher-dimensional automata. The development in this section does not require finiteness of c-sets.

\begin{definition}
We call the set $\Omega=\{0,\ast,1\}$ \emph{activity set}.
\end{definition}

The elements of $\Omega$ indicate the activity of concurrent events: $0$, $\ast$ and $1$ indicate that an event is not yet active, active and no longer active, respectively.

\begin{definition}
\label{d:concset-map}
Let $U,V,W$ be c-sets.
\begin{enumerate}
    \item A \emph{c-map} $f\colon V\to U$ is a pair of an injection $f\colon V\rightarrow U$ and an \emph{activity map}
  $\varepsilon_f\colon U\to \Omega$ such that $\varepsilon_f^{-1}(\ast)=f(U)$.

  \item Composition of c-maps $g\colon W\to V$ and $f\colon V\to U$ is given by
$fg\colon W\rightarrow U$ on injections and
\begin{equation}\label{e:ConcsetComposition}
  \varepsilon_{fg}(u)=
  \begin{cases}
  \varepsilon_g(f^{-1}(u)) & \text{if } u\in f(V),\\
    \varepsilon_f(u) & \text{otherwise}.
    \end{cases}
\end{equation}
\end{enumerate}
\end{definition}
It is routine to check that c-map composition is associative and that the identity functions $\id_U:U\to U$, which satisfy $\varepsilon_{\id_U}(U)=\{\ast\}$, form units of composition. 

Figure~\ref{fig:concset-map-comp} gives a pictorial account of c-maps and their composition. 

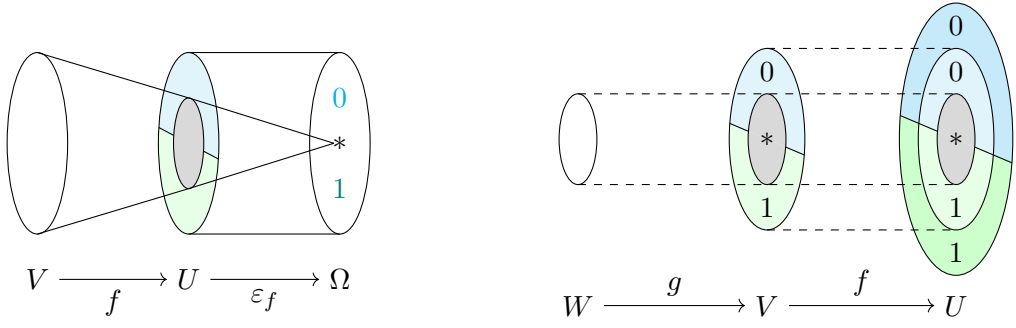
\begin{figure}[t]
  \begin{center}

      \begin{tikzpicture}[x=.4cm, y=.4cm]
  \draw (0,0)  ellipse (1 and 3);
  
  \fill[cyan!10, shift={(5,0)}] (-10:1 and 3) arc (-10:170:1 and 3) -- (170:0.5 and 1.5) arc (170:-10:0.5 and 1.5) -- cycle;
  \fill[green!10, shift={(5,0)}] (170:1 and 3) arc (170:350:1 and 3) -- (350:0.5 and 1.5) arc (350:170:0.5 and 1.5) -- cycle;
  
  \draw[fill=gray!30] (5,0) ellipse (0.5 and 1.5);
  \draw (5,0) ellipse (1 and 3);
  
  \draw                (10,0) ellipse (1 and 3);

  \draw (0,3) -- (5,1.5)   (0,-3) -- (5,-1.5);
  \draw   (5,3) -- (10,3)    (5,-3) -- (10,-3)
                (5,1.5) -- (9.8,0) (5,-1.5) -- (9.8,0);

  \draw (5,0) +(170:0.5 and 1.5) -- +(170:1 and 3);
  \draw (5,0) +(-10:0.5 and 1.5) -- +(-10:1 and 3);

  \node (V)     at (0,-4.5)  {$V$};
  \node (U)     at (5,-4.5)  {$U$};
  \node (Omega) at (10,-4.5) {$\Omega$};

  \draw[->] (V) -- node[below] {$f$} (U);
  \draw[->] (U) -- node[below]   {$\varepsilon_f$} (Omega);

  \node at (10,0)    {$\ast$};
  \node[cyan] at (10,1.5)  {$0$};
  \node[teal] at (10,-1.5) {$1$};
\end{tikzpicture}
\qquad\qquad\qquad
\begin{tikzpicture}[x=.5cm, y=.4cm]
  \fill[white] (0,0) +(0:0.5 and 1.5) arc (0:180:0.5 and 1.5) -- cycle;
  \fill[white] (0,0) +(180:0.5 and 1.5) arc (180:360:0.5 and 1.5) -- cycle;
  \draw (0,0) ellipse (0.5 and 1.5);

  \fill[gray!30, shift={(5,0)}] (-10:0.5 and 1.5) arc (-10:170:0.5 and 1.5) -- cycle;
  \fill[gray!30, shift={(5,0)}]  (170:0.5 and 1.5) arc (170:350:0.5 and 1.5) -- cycle;
  
  \fill[cyan!10, shift={(5,0)}] (-10:1 and 3) arc (-10:170:1 and 3) -- (170:0.5 and 1.5) arc (170:-10:0.5 and 1.5) -- cycle;
  \fill[green!10, shift={(5,0)}]  (170:1 and 3) arc (170:350:1 and 3) -- (350:0.5 and 1.5) arc (350:170:0.5 and 1.5) -- cycle;

  \draw (5,0) ellipse (0.5 and 1.5);
  \draw (5,0) ellipse (1 and 3);

  \fill[gray!30, shift={(10,0)}] (-10:0.5 and 1.5) arc (-10:170:0.5 and 1.5) -- cycle;
  \fill[gray!30, shift={(10,0)}]  (170:0.5 and 1.5) arc (170:350:0.5 and 1.5) -- cycle;

  \fill[cyan!10, shift={(10,0)}] (-10:1 and 3) arc (-10:170:1 and 3) -- (170:0.5 and 1.5) arc (170:-10:0.5 and 1.5) -- cycle;
  \fill[green!10, shift={(10,0)}]  (170:1 and 3) arc (170:350:1 and 3) -- (350:0.5 and 1.5) arc (350:170:0.5 and 1.5) -- cycle;

  \fill[cyan!20, shift={(10,0)}] (-10:1.5 and 4.5) arc (-10:170:1.5 and 4.5) -- (170:1 and 3) arc (170:-10:1 and 3) -- cycle;
  \fill[green!20, shift={(10,0)}]  (170:1.5 and 4.5) arc (170:350:1.5 and 4.5) -- (350:1 and 3) arc (350:170:1 and 3) -- cycle;

  \draw (10,0) ellipse (0.5 and 1.5);
  \draw (10,0) ellipse (1 and 3);
  \draw (10,0) ellipse (1.5 and 4.5);

  \draw[dashed] (0,-1.5) -- (10,-1.5)   (0,1.5) -- (10,1.5)
                (5,-3)   -- (10,-3)     (5,3)   -- (10,3);

  \draw (5,0) +(170:0.5 and 1.5) -- +(170:1 and 3);
  \draw (5,0) +(-10:0.5 and 1.5) -- +(-10:1 and 3);

  \draw (10,0) +(170:0.5 and 1.5) -- +(170:1.5 and 4.5);
  \draw (10,0) +(-10:0.5 and 1.5) -- +(-10:1.5 and 4.5);

  \node (W) at (0,-5.5)  {$W$};
  \node (V) at (5,-5.5)  {$V$};
  \node (U) at (10,-5.5) {$U$};

  \draw[->] (W) -- node[above] {$g$} (V);
  \draw[->] (V) -- node[above] {$f$} (U);

  \node at (5,0)      {$\ast$};
  \node at (5,2.25)   {$0$};
  \node at (5,-2.25)  {$1$};

  \node at (10,0)     {$\ast$};
  \node at (10,2.25)  {$0$};
  \node at (10,-2.25) {$1$};
  \node at (10,3.75)  {$0$};
  \node at (10,-3.75) {$1$};
\end{tikzpicture}
\end{center}
   \caption{A c-map $f\colon V\to U$ on the left; its composition with the c-map $g\colon W\to V$ on the right. In the diagram on the left, the concurrent events in $f(V)$ are active (labelled with $\ast$), those in the upper blue part of $U\setminus f(V)$ not yet active (labelled with $0$) and those in the lower green part no longer active (labelled with $1$). In the composite c-map $fg$ on the right, the activity of events in $f(V)$ is determined by the matching events in $V$;  that of events in $U\setminus f(V)$ is determined by the activity map $\varepsilon_f$.}
  \label{fig:concset-map-comp}
\end{figure}

\begin{definition}
The d-category $\CSet$ has c-sets as objects and c-maps as morphisms. Its d-structure is given, for all c-sets $U$, $V$, by 
\begin{align*}
\CSet^+(V,U) &= \{f\colon V\to U \mid \varepsilon^{-1}_f(1)=\emptyset\},\\
\CSet^-(V,U) &= \{f\colon V\to U \mid \varepsilon^{-1}_f(0)=\emptyset\}.
\end{align*}
\end{definition}
There is an obvious forgetful functor $\Forget\colon\cSet\to\Set$.
Isomorphisms in $\CSet$ are exactly those morphisms that are both for- and backmorphisms. Their underlying injections are bijections, and their activity maps have image $\{*\}$.

In the sequel, we overload notation for c-maps and their underlying injections, referring to the underlying category $\cSet$ or $\Set$ whenever confusion may arise.

In the following lemma we collect some basic properties of c-maps. 
\begin{lemma}
\label{l:cset-properties}
    Let $g\colon W\to V$ and $f\colon V\to U$ be c-maps.
\begin{enumerate}[1.]
\item \label{i:comp}
    The activity map $\varepsilon_{fg}$ is the unique function for which the following diagram commutes and, equivalently, the following two equations hold:

\begin{minipage}{.45\textwidth}
\begin{equation*}
\begin{tikzcd}
    V \ar[r,"f"] \ar[dr,"\varepsilon_g" swap]
    &
    U \ar[d,"\varepsilon_{fg}"]
    &
    U\setminus f(V) \ar[dl,"\varepsilon_f"]\ar[hookrightarrow, l, swap]
\\
    & \Omega &
\end{tikzcd}
\end{equation*}
\end{minipage}
\begin{minipage}{.5\textwidth}
\begin{align}
   \varepsilon_g &= \varepsilon_{fg}\circ f,\label{e:ConcsetComp1}\\
   \varepsilon_f|_{U\setminus f(V)} &=\varepsilon_{fg}|_{U\setminus f(V)}.\label{e:ConcsetComp2}
\end{align} 
\end{minipage}
\item \label{i:pure}
    If the c-map $fg$ is a formorphism (backmorphism), then so are $g$ and $f$.
\item \label{i:monic}
    All c-maps are monic.
\item \label{i:unique-ext}
    The injection $f$ extends uniquely to a formorphism and a backmorphism, setting $\varepsilon_f(u)$ to $0$ and to $1$ for all $u\in U\setminus f(V)$, respectively.
\end{enumerate}
\end{lemma}
\begin{proof}
The equations and the diagram in (1) rewrite the clauses for $\varepsilon_{fg}$ in \eqref{e:ConcsetComposition}. Item (2) is an immediate consequence of these clauses.
For (3), assume that the c-maps $f\colon V\to U$, $g,g'\colon W\to V$ satisfy $fg=fg'$ in $\cSet$. Then, in $\Set$, they also satisfy $fg=fg'$, which implies $g=g'$. Moreover, $\varepsilon_g = \varepsilon_{fg}\circ f = \varepsilon_{fg'}\circ f = \varepsilon_{g'}$ using \eqref{e:ConcsetComp1}.  Thus $g=g'$ in $\cSet$. Finally, (4) is immediate from the definitions.
\end{proof}

Next we record some elementary properties of pushouts of sets, which are frequently used in the sequel. We frequently use the diagram
\begin{equation}        
    \label{e:injective-pushout-square}
    \begin{tikzcd}
    W
        \ar[r,"g"]
        \ar[d,"f" swap]
    &
    V
        \ar[d, "j"]
    \\
    U\ar[r,"i" swap]
    &
    X
    \end{tikzcd}   
\end{equation}

\begin{lemma}
\label{l:pushout-square-desc}
Assume that \eqref{e:injective-pushout-square} is a pushout square in $\Set$ such that $f$ and $g$ are injections. Then
\begin{enumerate}[(1)]
\item the functions $i$ and $j$ are injective,
\item the square \eqref{e:injective-pushout-square} is a pullback diagram,
    \item there is a partition
    $X = X_U \sqcup X_W \sqcup X_V$, where 
    \begin{itemize}
        \item $X_U=X\setminus j(V)=i(U)\setminus j(V)=i(U\setminus f(W))$,
        \item $X_W=i(U)\cap j(V) = if(W) = jg(W)$,
        \item $X_V=X\setminus i(U)=j(V)\setminus i(U)=j(V\setminus g(W))$.
    \end{itemize}
\end{enumerate}
\end{lemma}

An \emph{injective pushout square} is a square satisfying the assumptions of Lemma \ref{l:pushout-square-desc} 

\begin{lemma}
\label{l:inj-pushout-crit}
    Assume that \eqref{e:injective-pushout-square} is a commutative square of injections. Then it is a pushout square if and only if $X=i(U)\cup j(V)$ and the induced map $W\to i(U)\cap j(V)$ is a bijection.
\end{lemma}

Next we collect necessary and sufficient conditions for activity maps in commutative squares of c-maps.

\begin{lemma}
\label{l:cset-square-commutes}
Suppose \eqref{e:injective-pushout-square} is a square of c-maps that is a pushout in $\Set$. Then the following conditions are equivalent:
\begin{enumerate}
\item
    $\varepsilon_g=\varepsilon_i\circ  j$ and $\varepsilon_f=\varepsilon_j\circ i$,
\item
    $\varepsilon_i(x)=
    \begin{cases}
        \varepsilon_g(j^{-1}(x)) & \text{if } x\in j(V),\\
        * & \text{if } x\in i(U)
    \end{cases}
    $\ \
    and\ \
    $\varepsilon_j(x)=
    \begin{cases}
        \varepsilon_f(i^{-1}(x)) & \text{if } x\in i(U),\\
        * & \text{if } x\in j(V),
    \end{cases}$   
\item
the square \eqref{e:injective-pushout-square} commutes in $\CSet$.
\end{enumerate}
\end{lemma}
\begin{proof}
The equivalence between (1) and (2) is immediate from the definition of c-maps and injectivity of $i$ and $j$. 
For that between (2) and (3) we use the partition in Lemma \ref{l:pushout-square-desc} to obtain
 \begin{equation*}
     \varepsilon_{if}(x)=
     \begin{cases}
        \varepsilon_f(i^{-1}(x)) & \text{for } x\in X_U,\\
        \rng & \text{for } x\in X_W,\\
        \varepsilon_i(x) & \text{for } x\in X_V,
     \end{cases}
     \qquad
     \varepsilon_{jg}(x)=
     \begin{cases}
        \varepsilon_j(x) & \text{for $x\in X_U$},\\
        \rng & \text{for $x\in X_W$},\\
        \varepsilon_g(j^{-1}(x)) & \text{for $x\in X_V$}.
     \end{cases}
 \end{equation*}
 Since $i(U)=X_U\cup X_W$ and $j(V)=X_V \cup X_W$, we have $\varepsilon_{if}=\varepsilon_{jg}$ if and only if (2) holds.
\end{proof}

\begin{lemma}
\label{l:dsquares-are-injective-pushouts}
The forgetful functor $\cSet\to\Set$ defines a bijection between d-squares in $\cSet$ and injective pushout squares in $\Set$.
\end{lemma}
\begin{proof}
Suppose \eqref{e:injective-pushout-square} is a d-square of c-maps with formorphisms $i$ and $g$ and backmorphisms $j$ and $f$. For every $x\in X\setminus(i(U)\cup j(V))$ we have
    $
        0=\varepsilon_i(x)=\varepsilon_{if}(x)=\varepsilon_{jg}(x)=\varepsilon_j(x)=1
    $. Thus, $X=i(U)\cup j(V)$.
    Furthermore, for every $x\in i(U)\cap j(V)$ we have $0\neq \varepsilon_f(i^{-1}(x))=\varepsilon_{if}(x)=\varepsilon_{jg}(x)=\varepsilon_g(j^{-1}(x))\neq 1$, which implies that $x\in \varepsilon_{if}^{-1}(*)=if(W)$. 
    The square \eqref{e:injective-pushout-square} is thus an injective pushout square in $\Set$ by Lemma \ref{l:inj-pushout-crit}.

Suppose \eqref{e:injective-pushout-square} is an injective pushout square. We equip the injections $f$, $g$, $i$ and $j$ with activity maps as in Lemma \ref{l:cset-properties}.\ref{i:unique-ext}, so that $g$ and $i$ are formorphisms, and $f$ and $j$ backmorphisms.
A calculation as for Lemma \ref{l:cset-square-commutes} shows that 
    \[
     \varepsilon_{if}(x)=\varepsilon_{jg}(x)=
     \begin{cases}
        1 & \text{for } x\in X_U,\\
        \rng & \text{for } x\in X_W,\\
        0 & \text{for } x\in X_V,
     \end{cases}
    \]
     so that the resulting d-square commutes (using notation from Lemma \ref{l:pushout-square-desc}).
\end{proof}

Figure~\ref{fig:factorisation-square} shows a d-square in $\CSet$. 

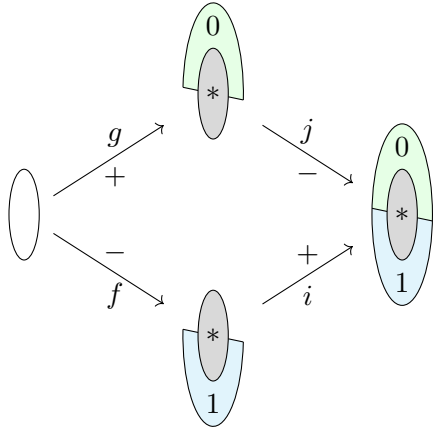
\begin{figure}[t]
    \centering

\begin{tikzpicture}[x=.5cm, y=.4cm]
  \draw (0,0) ellipse (0.4 and 1.5);

  \fill[green!10]       (5,4) +(-4:0.8 and 3) arc[start angle=-4, end angle=176, x radius=0.8, y radius=3] -- cycle;
  \draw                (5,4) +(-4:0.8 and 3) arc[start angle=-4, end angle=176, x radius=0.8, y radius=3];
  \draw[fill=gray!30] (5,4) ellipse (0.4 and 1.5);

  \fill[cyan!10]       (5,-4) +(176:0.8 and 3) arc[start angle=176, end angle=356, x radius=0.8, y radius=3] -- cycle;
  \draw                (5,-4) +(176:0.8 and 3) arc[start angle=176, end angle=356, x radius=0.8, y radius=3];
  \draw[fill=gray!30] (5,-4) ellipse (0.4 and 1.5);

  \fill[green!10] (10,0) +(-4:0.8 and 3) arc[start angle=-4, end angle=176, x radius=0.8, y radius=3] -- cycle;
  \fill[cyan!10]  (10,0) +(176:0.8 and 3) arc[start angle=176, end angle=356, x radius=0.8, y radius=3] -- cycle;
  \draw (10,0) ellipse (0.8 and 3);
  \draw[fill=gray!30] (10,0) ellipse (0.4 and 1.5);

  \draw (5,4)   +(176:0.4 and 1.5) -- +(176:0.8 and 3)
        (5,4)   +(-4:0.4 and 1.5)  -- +(-4:0.8 and 3);

  \draw (5,-4)  +(176:0.4 and 1.5) -- +(176:0.8 and 3)
        (5,-4)  +(-4:0.4 and 1.5)  -- +(-4:0.8 and 3);

  \draw (10,0)  +(176:0.4 and 1.5) -- +(176:0.8 and 3)
        (10,0)  +(-4:0.4 and 1.5)  -- +(-4:0.8 and 3);

  \node (B) at (0,0) {};
  \node (C) at (5,4)  {$\ast$};
  \node (D) at (5,-4) {$\ast$};
  \node (E) at (10,0) {$\ast$};

  \node at (5,6.25)   {$0$};
  \node at (5,-6.25)  {$1$};
  \node at (10,2.25)  {$0$};
  \node at (10,-2.25) {$1$};

  \draw[->, shorten >=.5cm, shorten <=.3cm] (B) -- node[above] {$g$} node[below] {$+$} (C);
  \draw[->, shorten >=.5cm, shorten <=.3cm] (B) -- node[below] {$f$} node[above] {$-$} (D);
  \draw[->, shorten >=.5cm, shorten <=.5cm] (C) -- node[above] {$j$} node[below] {$-$} (E);
  \draw[->, shorten >=.5cm, shorten <=.5cm] (D) -- node[below] {$i$} node[above] {$+$} (E);
\end{tikzpicture}

    \caption{A d-square in $\CSet$ with values of the activity maps of $f$, $g$ and  $jg=if$. Intuitively, in such a d-square, the same events are active in the images of $f$ and $g$ (up to renaming), while the images of $f$ and $g$ contain no events that are no longer active and not yet active, respectively. In the rightmost set, all these events are combined consistently (up to renaming). If all injections in the diagram are inclusions, consistency means in particular that the rightmost set is the union of the upper and the lower set, while the leftmost set is the intersection of the upper and the lower set.}
    \label{fig:factorisation-square}    
\end{figure}

\begin{lemma}
    \label{l:dsquare-is-dpushout}
    Every d-square in $\cSet$ is a d-pushout.
\end{lemma}
\begin{proof}
    Consider the following d-square in $\cSet$:
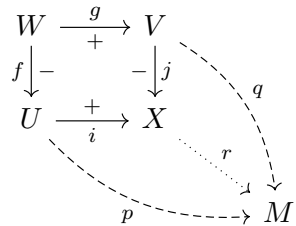
    \[
    \begin{tikzcd}
    W\ar[r,"g", "+" swap]\ar[d, "f" swap, "-"] & V\ar[d,"j", "-" swap]\ar[ddr,dashed,"q", bend left=25] \\
    U\ar[r, "i" swap,"+"]\ar[drr,dashed,"p" swap, bend right=25] & X\ar[dr,dotted,"r"] \\
    & & M
\end{tikzcd} 
    \]
    Let $M$ be a c-set and let the c-maps $p\colon U\to M$ and $q\colon V\to M$ satisfy $pf=qg$. The underlying diagram in $\Set$ is a pushout by Lemma \ref{l:dsquares-are-injective-pushouts}, which provides a unique function $r\colon X\to M$ such that $ri=p$ and $rj=q$ in $\Set$. We must show that it extends uniquely to a c-map.

For injectivity of $r$ suppose $x\neq x'\in X$. There are three cases: 
\begin{itemize}
    \item If $x,x'\in i(U)$, then $r(x)\neq r(x')$ by injectivity of $i^{-1}|_{i(U)}$ and $p$. 
    \item For $x,x'\in j(V)$ there is a similar proof.
    \item If $x\in i(U)\setminus j(V)$ and $x'\in j(V)\setminus i(U)$, then $\varepsilon_j(x)=1$ and $\varepsilon_i(x')=0$. Hence 
    \[
    \varepsilon_{rj}(r(x))= \varepsilon_j(x)= 1 \neq 0 =\varepsilon_i(x') = \varepsilon_{rj}(r(x'))
    \]
by \eqref{e:ConcsetComp1} and thus $r(x)\neq r(x')$.
\end{itemize}
Next we define an activity map $\varepsilon_r\colon X\to \Omega$ by $\varepsilon_r(m)=\rng$ for $m\in r(X)$ and 
\[
    \varepsilon_r(m)=
    \varepsilon_p(m)=
    \varepsilon_{pf}(m)=
    \varepsilon_{qj}(m)=
    \varepsilon_q(m)
\]
for $m\notin r(X)$. All equalities in this definition follow by \eqref{e:ConcsetComp2}.
 Now we show that $p=ri$ using \ref{l:cset-properties}.\ref{i:comp}.
First, $\varepsilon_r|_{M\setminus r(X)}=\varepsilon_p|_{M\setminus r(X)}$ follows from the definition of $\varepsilon_r$. It remains to prove that $\varepsilon_i(x)=\varepsilon_p(r(x))$ for $x\in X$. If $x\in i(U)$, then $\varepsilon_i(x)=*=\varepsilon_p(r(x))$ because $r(x)\in p(X)$. If $x\in X\setminus i(X)$, then
\[
    \varepsilon_i(x)=
    \varepsilon_{if}(x)=
    \varepsilon_{jg}(x)=
    \varepsilon_g(j^{-1}(x))=
    \varepsilon_{qg}(q(j^{-1}(x)))=
    \varepsilon_{qg}(r(x)))=
    \varepsilon_{pf}(r(x)))=
    \varepsilon_{p}(r(x))).
\]
Finally, $\varepsilon_r$ is uniquely determined by $\varepsilon_p$ owing to \eqref{e:ConcsetComp2}.
\end{proof}

\begin{proposition}
\label{p:cset-is-dcocomplete}
The forgetful functor $\Forget:\cSet\to\Set$ creates d-pushouts. The category $\cSet$ is therefore d-cocomplete.
\end{proposition}
\begin{proof}
We adopt notation from \eqref{e:pushout-creation}.
Let $X\colon\Lambda\to \cSet$ be a d-span and let $Y\colon\lozenge\to \Set$ by a pushout extending $\Forget\circ X$. Since $Y$ is an injective pushout square, it lifts uniquely to a d-square $Y'\colon \lozenge\to \cSet$ by Lemma \ref{l:dsquares-are-injective-pushouts}, which is a d-pushout by Lemma \ref{l:dsquare-is-dpushout}. Moreover, $Y'|_{\Lambda}=Y$ by Lemma \ref{l:cset-properties}.\ref{i:unique-ext}. Thus $\Forget$ creates d-pushouts. The second claim follows from completeness of $\Set$.
\end{proof}

\begin{proposition}
\label{p:cset-pullback-creation}
    The forgetful functor $\cSet\to\Set$ creates pullbacks.
\end{proposition}
\begin{proof}
Assume that the following square in $\cSet$ is a pullback in $\Set$:
\[
\begin{tikzcd}
    M
        \ar[drr,bend left=25,"t",dashed]
        \ar[ddr,bend right=25,"s" swap ,dashed]
        \ar[dr,"r",dotted]    
    \\
    & W
        \ar[r,"g"]
        \ar[d, "f" swap]
    & V
        \ar[d,"j"]
    \\
    & U
        \ar[r, "i" swap]
    & X
\end{tikzcd}   
\]
Let $M$ be a c-set, let $s:M\to U$ and  $t:M\to V$ be c-maps satisfying $is=jt$. Since the square is a pullback in $\Set$, there exists a unique function $r:M\to W$ such that $s=fr$ and $t=gr$ in $\Set$. 

Define $\varepsilon_r=\varepsilon_s\circ f$. For $u\in U\setminus f(W)$ we have 
\[
    \varepsilon_f(u)=
    \varepsilon_{if}(i(u))=
    \varepsilon_{jg}(i(u))=
    \varepsilon_j(i(u))=
    \varepsilon_{jt}(i(u))=
    \varepsilon_{is}(i(u))=
    \varepsilon_s(u)
\]
since $i(u)\notin j(V)$.
Using Lemma \ref{l:cset-properties}.\ref{i:comp}, this shows that $s=fr$ in $\cSet$. Moreover, uniqueness is enforced by \eqref{e:ConcsetComp1}.
Since $\varepsilon_r=\varepsilon_s\circ f=\varepsilon_{is}\circ if =\varepsilon_{jt}\circ jg =\varepsilon_t\circ g$, a similar argument shows that $t=gr$ in $\cSet$.
\end{proof}

\begin{lemma}
\label{l:dpushouts-are-pullbacks}
Every d-pushout in $\cSet$ is a pullback.
\end{lemma}
\begin{proof}
Let $X\colon \lozenge\to \cSet$ be a d-pushout. The composition $\Forget^{\cSet}_\Set\circ X$ is an injective pushout square, and hence a pullback in $\Set$ by Lemma \ref{l:pushout-square-desc}. Thus $X$ is a pullback by Proposition \ref{p:cset-pullback-creation}.
\end{proof}

\begin{remark}
The span $\{1\}\forleftarrow{} \emptyset\forrightarrow{} \{1\}$ does not have a colimit in $\CSet$. Hence $\CSet$ does not have all pushouts. Furthermore, there is no initial object in $\cSet$. Hence arbitrary limits do not exist in this category.
\end{remark}

\section{Concurrent posets}
\label{s:c-posets}

In this section, we first introduce concurrent posets (c-posets) that generalise concurrent sets. They are meant to represent sets of events together with their temporal precedences. Then we construct directed pushouts of c-posets and show that they are preserved by the restricted Yoneda functor to the category of presheaves over c-sets. This result is used in the next section to verify that c-posets form a directed cocompletion of c-sets.

We now equip the activity set $\Omega$ with the partial order $\le_\Omega$ induced by
$1 <_\Omega \ast <_\Omega 0$.

\begin{definition}
A \emph{c-poset} $(P,<)$ consists of a c-set $P$ and a strict partial order $<$ on $P$ (an irreflexive transitive relation) called \emph{precedence}.
\end{definition}

We write $\parallel$ for the incomparability relation with respect to $\le $,
that is, ${\parallel} =({\not<} \cap {\not>} \cap {\neq})$. Elements $p,p'\in P$ such that $p\parallel p'$ are called \emph{parallel}.

By definition, c-posets are finite posets. Yet the entire development in this section does not require finiteness.

\begin{definition} \label{d:PosetMap}~
 A \emph{c-poset map} between c-posets $Q$ and $P$ is a c-map $f\colon Q\to P$
 such that
 \begin{enumerate}
\item its injection $f\colon Q\to P$ reflects precedence ($f(q)<_Pf(q')\implies q<_Q q'$),
\item its activity map $\varepsilon_f\colon P\to \Omega$ is weakly order-preserving ($p<_Pp'\implies \varepsilon_f(p)\leq_\Omega \varepsilon_f(p')$).
 \end{enumerate}
 Composition of c-poset maps is determined by the underlying c-maps.
\end{definition}

See Figure \ref{f:cposet-map} for an example of a c-poset map. 

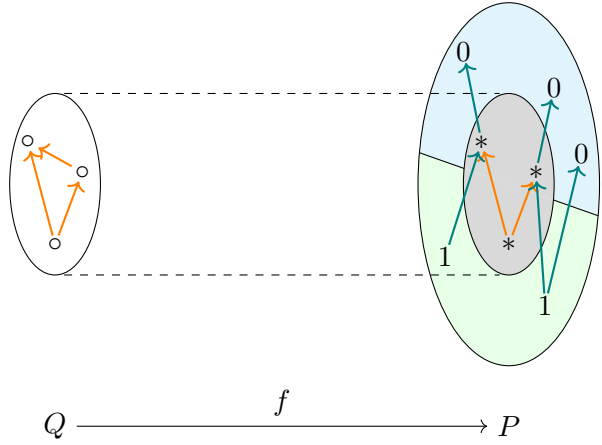
\begin{figure}
    \centering
\begin{tikzpicture}[x=1.2cm, y=.8cm]
  \draw                (0,0) ellipse (0.5 and 1.5);

  \fill[gray!30, shift={(5,0)}] (-10:0.5 and 1.5) arc (-10:170:0.5 and 1.5) -- cycle;
  \fill[gray!30, shift={(5,0)}]  (170:0.5 and 1.5) arc (170:350:0.5 and 1.5) -- cycle;

  \fill[cyan!10, shift={(5,0)}] (-10:1 and 3) arc (-10:170:1 and 3) -- (170:0.5 and 1.5) arc (170:-10:0.5 and 1.5) -- cycle;
  \fill[green!10, shift={(5,0)}]  (170:1 and 3) arc (170:350:1 and 3) -- (350:0.5 and 1.5) arc (350:170:0.5 and 1.5) -- cycle;

  \draw (5,0) ellipse (1 and 3);
  \draw (5,0) ellipse (0.5 and 1.5);

  \draw[dashed] (0.1,-1.5) -- (4.9,-1.5)
                (0.1,1.5)  -- (4.9,1.5);

  \draw (5,0) +(170:0.5 and 1.5) -- +(170:1 and 3);
  \draw (5,0) +(-10:0.5 and 1.5) -- +(-10:1 and 3);

  \node (W) at (0,-4) {$Q$};
  \node (V) at (5,-4) {$P$};
  \draw[->] (W) -- node[above] {$f$} (V);

  \node (a0) at (0,-1)    {$\circ$};
  \node (b0) at (0.3,0.2) {$\circ$};
  \node (c0) at (-0.3,0.7){$\circ$};

  \draw[->, orange, thick, shorten >=-3pt, shorten <=-3pt] (a0) -- (b0);
  \draw[->, orange, thick, shorten >=-3pt, shorten <=-3pt] (a0) -- (c0);
  \draw[->, orange, thick, shorten >=-3pt, shorten <=-3pt] (b0) -- (c0);

  \node (a) at (5,-1)    {$\ast$};
  \node (b) at (5.3,0.2) {$\ast$};
  \node (c) at (4.7,0.7) {$\ast$};
  \node (d) at (5.4,-2)  {$1$};
  \node (e) at (5.8,0.5)  {$0$};
  \node (f) at (5.5,1.6)  {$0$};
  \node (g) at (4.3,-1.2) {$1$};
  \node (h) at (4.5,2.2)  {$0$};

  \draw[->, orange, thick, shorten >=-3pt, shorten <=-3pt] (a) -- (b);
  \draw[->, orange, thick, shorten >=-3pt, shorten <=-3pt] (a) -- (c);

  \draw[->, teal, thick, shorten >=-3pt, shorten <=-3pt] (d) -- (b);
  \draw[->, teal, thick, shorten >=-3pt, shorten <=-3pt] (d) -- (e);
  \draw[->, teal, thick, shorten >=-3pt, shorten <=-3pt] (b) -- (f);
  \draw[->, teal, thick, shorten >=-3pt, shorten <=-3pt] (c) -- (h);
  \draw[->, teal, thick, shorten >=-3pt, shorten <=-3pt] (g) -- (c);
\end{tikzpicture}
    \caption{A c-poset map $f\colon Q\to P$. As the precedence is order-reflecting, there cannot be fewer orange arrows on the left than on the right. As the activity map is weakly order-preserving, there cannot be any teal arrows from $0$ to $\ast$ or $1$ or from $\ast$ to $1$.}
    \label{f:cposet-map}
\end{figure}

\begin{lemma}
\label{l:cposet-maps-closed-under-composition}
    c-poset maps are closed under composition.
\end{lemma}
\begin{proof}
Order-reflecting injections are closed under composition. For c-poset maps
  $g\colon R\to Q$ and $f\colon Q\to P$ it remains to check that $\varepsilon_{fg}$ is weakly order-preserving. Suppose $p<_P p'$. There are three cases:
  \begin{itemize}
  \item If $p,p'\in f(Q)$, then $f(q) = p <_P p' = f(q')$ for some $q,q'\in Q$ and thus $q<_Q q'$ by order-reflection. Then $\varepsilon_{fg}(f(q))= \varepsilon_g(q)\le_\Omega \varepsilon_g(q')=\varepsilon_{fg}(f(q')$ using \eqref{e:ConcsetComp1} and weak order-preservation of $\varepsilon_g$. 
  \item If $p\in P\setminus f(Q)$, then $\varepsilon_{fg}(p)=\varepsilon_f(p)\neq *$ by \eqref{e:ConcsetComp2} and the definition of c-maps. For $\varepsilon_f(p)=1$ there is nothing to
  show. Otherwise, $0=\varepsilon_f(p)\leq_\Omega\varepsilon_{f}(p')=0$ and thus $\varepsilon_{fg}(p')=0$ by \eqref{e:ConcsetComp2}.
  \item The case of $p'\in P\setminus f(Q)$ is symmetric.\qedhere
  \end{itemize}
\end{proof}

An injection $f:(Q,{\prec_Q})\to (P,{\prec_P})$ between relational structures is a \emph{relation embedding} if 
\[
q\prec_Qq'\iff f(q)\prec_Pf(q')
\]
for all $q,q'\in Q$.
A c-poset map is a \emph{precedence embedding} if its injection is a relation embedding.

\begin{definition}
\label{d:c-poset}
  The d-category $\cPos$ has c-posets as objects and c-poset maps as morphisms,
  \begin{itemize}
      \item formorphisms are precedence embeddings that are formorphisms of c-sets,
      \item backmorphisms are precedence embeddings that are backmorphisms of c-sets.
  \end{itemize}
\end{definition}

The forgetful functor $\cPos\to \CSet$ does not reflect the d-structure because for- and backmorphisms must preserve precedence.

Every c-set can be regarded as a c-poset with empty precedence; c-poset maps between c-sets are exactly c-maps.  We can thus regard $\CSet$ as a subcategory of $\cPos$.

\begin{lemma}
\label{l:xi-cpos-functors}
The category $\CSet$ is a full d-subcategory of $\cPos$. 
The forgetful d-functor $\Forget\colon\cPos\to\CSet$ is faithful and a left adjoint to the inclusion functor $\CSet\hookrightarrow\cPos$. Therefore, $\CSet$ is a reflective subcategory of $\cPos$.
\end{lemma}
\begin{proof}
If $U$ is a c-set and $P$ a c-poset, then every c-map $\Forget(P)\to U$ defines a (unique) c-poset map $P\to U$. Both conditions in Definition \ref{d:c-poset} are trivially satisfied in the absence of precedence in $U$. This shows adjointness, and the remaining statements are clear.
\end{proof}

Next we list basic properties of c-poset maps. A subset $A$ of poset $(P,<)$ is \emph{down-closed} if $p<a$ implies $p\in A$ for all $a\in A$ and $p\in P$. An \emph{up-closed} subset is down-closed in the opposite poset.
\begin{lemma}
\label{l:cposet-map-properties}
    Let $P$ and $Q$ be c-posets.
\begin{enumerate}
\item
    If $f\colon Q\to P$ is a c-poset map, then $\varepsilon^{-1}(1)$ is a down-closed and $\varepsilon^{-1}(0)$ an up-closed subset of $P$.
\item
    The c-poset maps $f\colon \emptyset\to P$ are in bijective correspondence with partitions of $P$ into a down-closed subset $\varepsilon_f^{-1}(1)$ and an up-closed subset $\varepsilon_f^{-1}(0)$.
\item
    Each inclusion $i\colon A\rightarrow P$ of a maximal antichain $A$ into $P$ extends uniquely to a c-poset map, defining 
    \[
    \varepsilon_i(p)=
    \begin{cases}
        1 & \text{if } p<a \text{ for some } a\in A,\\
        * & \text{if } p\in A,\\
        0 & \text{if } a<p \text{ for some } a\in A.
    \end{cases}
    \]
\item 
\label{l:min-incl}
    An injection  $f:Q\to P$ extends to a formorphism if and only if it is a precedence embedding and $f(Q)$ is down-closed. It extends to a backmorphism if and only if it is a precedence embedding and $f(Q)$ is up-closed.
\end{enumerate}    
\end{lemma}

Next we consider d-pushouts in $\cPos$.
For a function $f\colon Y\to X$ and a relation $\prec$ on $Y$, we define the relation $f({\prec})=\{(f(y),f(y'))\mid y\prec y'\}$ on $X$.

\begin{theorem}
\label{t:cpos-has-proper-pushouts}
 Let $(P\backleftarrow{f}R\forrightarrow{g} Q)$ be a d-span in $\cPos$ and $\Forget\colon \cPos\to \cSet$ the forgetful functor.
 \begin{enumerate}
 \item Every d-pushout
 \begin{equation*}
    \begin{tikzcd}[sep=30]
        \Forget(R)\ar[r,"\Forget(g)","+" swap]\ar[d,"\Forget(f)" swap,"-"] & \Forget(Q)\ar[d,"j","-" swap] \\
        \Forget(P)\ar[r,"i" swap,"+"] & S
    \end{tikzcd}
\end{equation*}
in $\CSet$ lifts to a d-pushout 
\begin{equation}
\label{e:cposet-concatenation}
    \begin{tikzcd}[sep=30]
        R\ar[r,"g","+" swap]\ar[d,"f" swap,"-"] & Q\ar[d,"j","-" swap] \\
        P\ar[r,"i" swap,"+"] & (S,{<_S})
    \end{tikzcd}
\end{equation}
in $\cPos$ with precedence 
\begin{equation}
\label{e:concat-precedence}
    {<_S} = {i({<_P})} \cup {j({<_Q})} \cup {i(P\setminus f(R)) \times j(Q\setminus g(R)))}.
\end{equation}
\item The category $\cPos$ is d-cocomplete, and $\Forget\colon \cPos\to \cSet$ preserves d-pushouts.
\end{enumerate}
\end{theorem}
\begin{proof}
For (1) we proceed by proving a sequence of claims.

\smallskip\noindent
(A)\; The set $j(Q)$ is an up-closed and the set $i(P)$ a down-closed  subset of $S$.\\
    Suppose $s<_Ss'$, where $s\in j(Q)$ and $s'\in S\setminus j(Q)$. Then $s'\in i(P)$ and, by \eqref{e:concat-precedence}, also $s\in i(P)$ and $i^{-1}(s)<_P i^{-1}(s')$. But $i^{-1}(s)\in f(R)$ and $i^{-1}(s')\not\in f(R)$. Since $f$ is a c-poset backmorphism,  $\varepsilon_f(i^{-1}(s))=*\not\leq_\Omega 1=\varepsilon_f(i^{-1}(s'))$ -- a contradiction.
    Thus $j(Q)$ is up-closed. The proof of the second claim is similar.

    \smallskip\noindent
    (B)\; The maps $i$ and $j$ are relation embeddings.\\
    Assume that $i(p)<_S i(p')$ with $p,p'\in P$. Using \eqref{e:concat-precedence} we have $p<_P p'$ or $j^{-1}(i(p))<_Q j^{-1}(i(p'))$. In the second case, $p,p'\in f(R)$ and then $f^{-1}(p)=g^{-1}(j^{-1}(i(p))<_R g^{-1}(j^{-1}(i(p'))=f^{-1}(p')$, since $g$ reflects precedence. Thus $p<_Pp'$, since $f$ preserves precedence. This shows that $i$ reflects $<$. As $i$ preserves $<$ by \eqref{e:concat-precedence}, it is a relation embedding. The proof for $j$ is similar.

    \smallskip\noindent
    (C)\; The relation $<_S$ is transitive.\\
    Suppose $s<_Ss'<_Ss''$. There are three cases:
    \begin{itemize}
    \item 
        If $s\in j(Q)$, then $s'\in j(Q)$ and further $s''\in j(Q)$ by (A). It follows from (B) that $j^{-1}(s)<_Q j^{-1}(s')<_Q j^{-1}(s'')$, which implies that $j^{-1}(s)<_Q j^{-1}(s'')$ by transitivity of $<_Q$. Therefore, $s=j(j^{-1}(s))<_S j(j^{-1}(s''))=s''$.
    \item 
        For $s''\in i(P)$, a similar argument shows that $s<_S s''$.
    \item 
        If $s\in S\setminus j(Q)=i(P\setminus f(R))$ and $s''\in S'\setminus i(P)=j(Q\setminus g(R))$, then $s<_S s''$ by \eqref{e:concat-precedence}.
    \end{itemize}

    \smallskip\noindent
    (D)\; The relational structure $(S,{<_S})$ is a c-poset.\\
    Irreflexivity of $<_S$ follows from (B) and irreflexivity of $<_P$ and $<_Q$; transitivity is shown in (C).

    \smallskip\noindent
    (E)\; Diagram \eqref{e:cposet-concatenation} is a d-square of c-posets.\\
    The diagram commutes in $\cPos$ because it commutes in $\cSet$. The c-map $i$ is a c-poset formorphism by Lemma \ref{l:cposet-map-properties}.4 because it is a precedence embedding by (B) and $i(P)$ is a down-closed subset of $S$ by (A). For similar reasons, $j$ is a c-poset backmorphism.

    \medskip
    It remains to show that \eqref{e:cposet-concatenation} is a pushout. Let $M$ be a c-poset and $u:P\to M$ and $v:Q\to M$ c-poset maps satisfying $uf=vg$. This yields the following commutative diagram of solid arrows in $\cPos$:
\begin{equation*}
  \begin{tikzcd}
R \ar[r,"g", "+" swap]\ar[d,"f" swap, "-"] & Q\ar[d,"j", 
 "-" swap]\ar[ddr, "v", bend left=25]&\\
P \ar[r, swap, "i", "+" swap]\ar[drr, swap, "u", bend right=25]& S\ar[dr, dashed, "t"]&\\
&& M
\end{tikzcd}
\end{equation*}
Since the d-square is a pushout in $\CSet$, there is a unique c-map $t\colon S\to M$ for which $ti=u$ and $tj=v$ in $\cSet$. We need to check that $t$ is a c-poset map. 

    \smallskip\noindent
    (F)\; The c-map $t$ reflects precedence.\\
    Suppose $t(s)<_Mt(s')$ for $s,s'\in S$. We show that $s<_Ss'$. There are four cases.
    \begin{itemize}
    \item
        If $s,s'\in i(P)$, then $i^{-1}(s)=u^{-1}(t(s))<_P u^{-1}(t(s'))=i^{-1}(s')$ since $u$ reflects precedence, and then $s<_Ss'$ since $i$ preserves precedence. 
    \item
        For $s,s'\in j(Q)$, the argument is similar.
    \item
        If $s\in i(P\setminus f(R))=i(P)\setminus j(Q)$ and $s'\in j(Q\setminus g(R))=j(Q)\setminus i(P)$, then $s<_Ss'$ by \eqref{e:cposet-concatenation}.
    \item
        If $s\in j(Q)\setminus i(P)$ and $s'\in i(P)\setminus j(Q)$, then $\varepsilon_{tif}(t(s))=\varepsilon_{if}(s)=\varepsilon_i(s)=0$ by identity \eqref{e:ConcsetComp1}, as $f$ is a formorphism and as $s\notin i(P)$. Similarly, $\varepsilon_{tjg}(t(s'))=1$. But $tif=uf=vg=tjg$ is a c-poset map, so that $t(s')<_M t(s)$ -- a contradiction.
    \end{itemize}

    \smallskip\noindent
    (G)\; The activity map $\varepsilon_t$ is weakly order-preserving.\\
    Assume that $m<_Mm'$, $m,m'\in M$. We show that $\varepsilon_t(m)\leq_\Omega \varepsilon_t(m')$. Again, there are several cases.
    \begin{itemize}
    \item
        If $m,m'\notin t(S)$, then $\varepsilon_t(m)=\varepsilon_u(m)\leq_\Omega \varepsilon_u(m')=\varepsilon_t(m')$ since $u=ti$ is a c-poset map.
    \item
        If $m,m'\in t(S)$, then $\varepsilon_t(m)=*=\varepsilon_t(m')$. 
    \item
        If $m\in t(i(P))$ and $m'\notin t(S)$, then $\varepsilon_t(m)=\varepsilon_u(m)=*$, $\varepsilon_t(m')=\varepsilon_u(m')$ and  $\varepsilon_u(m)\leq_\Omega\varepsilon_u(m')$, as $u$ is a c-poset map. 
    \item
        The remaining cases of $m\notin t(S)$ and $m'\in t(i(P))$,  $m\in t(j(Q))$ and $m'\notin t(X)$, as well as $m\notin t(X)$ and $m'\in t(j(Q))$ are similar.
\end{itemize}
Item (2) follows immediately from (1) and the fact that
$\cSet$ is d-cocomplete (Proposition \ref{p:cset-is-dcocomplete}).
\end{proof}

This construction generalises the gluing composition of posets in \cite{FJSZ_HDALang}.  See Figure~\ref{f:cposet-concat} for an example.

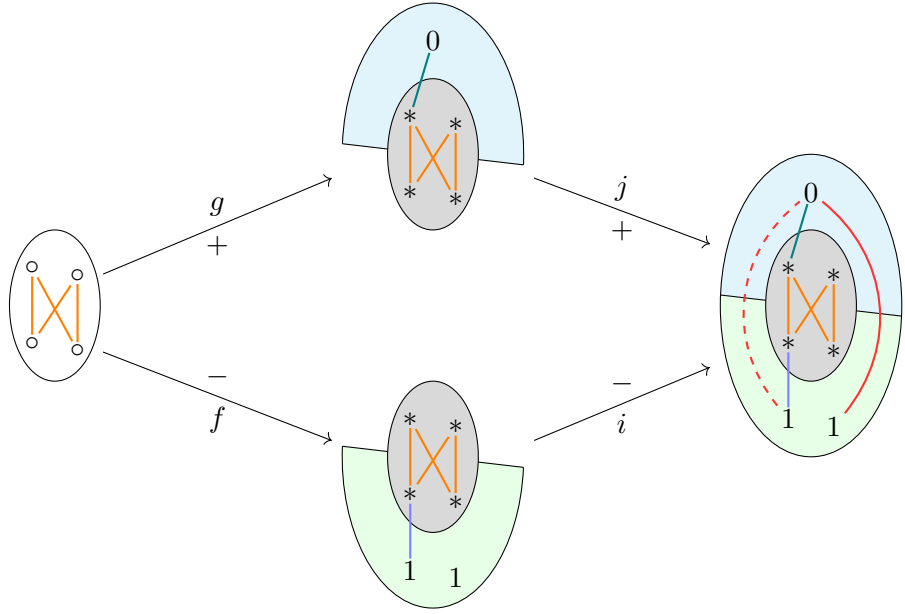
\begin{figure}
    \centering
\begin{tikzpicture}[node distance=2cm, shorten >=-3pt, shorten <=-3pt,on grid]

  \draw (3mm,5mm) ellipse (0.6 and 1);

  \path (5.3,2.5) +(-4:0.6 and 1)   coordinate (T_in_R);
  \path (5.3,2.5) +(-4:1.2 and 2)   coordinate (T_out_R);
  \path (5.3,2.5) +(176:0.6 and 1) coordinate (T_in_L);
  \path (5.3,2.5) +(176:1.2 and 2) coordinate (T_out_L);

  \path (5.3,-1.5) +(356:0.6 and 1) coordinate (B_in_R);
  \path (5.3,-1.5) +(356:1.2 and 2) coordinate (B_out_R);
  \path (5.3,-1.5) +(176:0.6 and 1) coordinate (B_in_L);
  \path (5.3,-1.5) +(176:1.2 and 2) coordinate (B_out_L);

  \path (10.3,0.5) +(-4:0.6 and 1)  coordinate (R_in_R);
  \path (10.3,0.5) +(-4:1.2 and 2)  coordinate (R_out_R);
  \path (10.3,0.5) +(176:0.6 and 1) coordinate (R_in_L);
  \path (10.3,0.5) +(176:1.2 and 2) coordinate (R_out_L);

  \fill[cyan!10] (T_out_R) arc[start angle=-4, end angle=176, x radius=1.2, y radius=2] -- cycle;
  \draw[shorten >=0pt, shorten <=0pt] (T_out_R) arc[start angle=-4, end angle=176, x radius=1.2, y radius=2];
  \draw[fill=gray!30] (5.3,2.5) ellipse (0.6 and 1);
  \draw[-, shorten >=0pt, shorten <=0pt] (T_in_R) -- (T_out_R);
  \draw[-, shorten >=0pt, shorten <=0pt] (T_in_L) -- (T_out_L);

  \fill[green!10] (B_out_L) arc[start angle=176, end angle=356, x radius=1.2, y radius=2] -- cycle;
  \draw[shorten >=0pt, shorten <=0pt] (B_out_L) arc[start angle=176, end angle=356, x radius=1.2, y radius=2];
  \draw[fill=gray!30] (5.3,-1.5) ellipse (0.6 and 1);
  \draw[-, shorten >=0pt, shorten <=0pt] (B_in_R) -- (B_out_R);
  \draw[-, shorten >=0pt, shorten <=0pt] (B_in_L) -- (B_out_L);

  \fill[cyan!10] (R_out_R) arc[start angle=-4, end angle=176, x radius=1.2, y radius=2] -- cycle;
  \fill[green!10]  (R_out_L) arc[start angle=176, end angle=356, x radius=1.2, y radius=2] -- cycle;
  \draw (10.3,0.5) ellipse (1.2 and 2);
  \draw[fill=gray!30] (10.3,0.5) ellipse (0.6 and 1);
  \draw[-, shorten >=0pt, shorten <=0pt] (R_in_R) -- (R_out_R);
  \draw[-, shorten >=0pt, shorten <=0pt] (R_in_L) -- (R_out_L);

  \node (a0) {$\circ$};
  \node (b0) [above=1cm of a0] {$\circ$};
  \node (c0) [below right=1mm and 6mm of a0] {$\circ$};
  \node (d0) [above=1cm of c0] {$\circ$};
  \path[-]
  (a0) edge[orange,thick] (b0)
  edge[orange,thick] (d0)
  (c0) edge[orange,thick] (b0)
  edge[orange,thick] (d0);

  \node (a1) [above right =2cm and 5cm of a0] {$\ast$};
  \node (b1) [above=1cm of a1] {$\ast$};
  \node (c1) [below right=1mm and 6mm of a1] {$\ast$};
  \node (d1) [above=1cm of c1] {$\ast$};
  \node (f11) [above right=1cm and 3mm of b1] {$0$};
  \path[-]
  (a1) edge[orange,thick] (b1)
  edge[orange,thick] (d1)
  (c1) edge[orange,thick] (b1)
  edge[orange,thick] (d1)
  (b1) edge[teal,thick] (f11);

  \node (a2) [below right =2cm and 5cm of a0] {$\ast$};
  \node (b2) [above=1cm of a2] {$\ast$};
  \node (c2) [below right=1mm and 6mm of a2] {$\ast$};
  \node (d2) [above=1cm of c2] {$\ast$};
  \node (d22) [below=1cm of a2] {$1$};
  \node (e2) [below=1cm of c2] {$1$};
  \path[-]
  (a2)edge[orange,thick] (b2)
  edge[orange,thick] (d2)
  (c2) edge[orange,thick] (d2)
  edge[orange,thick] (b2)
  (d22) edge[blue!50,thick] (a2);

  \node (a3) [above right =2cm and 5cm of a2] {$\ast$};
  \node (b3) [above=1cm of a3] {$\ast$};
  \node (c3) [below right=1mm and 6mm of a3] {$\ast$};
  \node (d3) [above=1cm of c3] {$\ast$};
  \node (dd3) [below=1cm of c3] {$1$};
  \node (e3) [below=1cm of a3] {$1$};
  \node (f3) [above right=1cm and 3mm of b3] {$0$};
  \path[-]
  (a3)edge[orange,thick] (b3)
  edge[orange,thick] (d3)
  (c3) edge[orange,thick] (d3)
  edge[orange,thick] (b3)
  (e3) edge[blue!50,thick] (a3)
  (b3) edge[teal,thick] (f3)
  (e3) edge[red!80,thick,dashed,bend left=45] (f3)
  (dd3) edge[red!80,thick,bend right=45] (f3);

  \node (g1) [right= 3mm of d0] {};
  \node (g2) [below left=8mm and 10mm of b1] {};
  \node (f1) [right= 3mm of c0] {};
  \node (f2) [below left= 3mm and 10mm of b2] {};
  \node (j1) [above right=3mm and 10mm of c1] {};
  \node (j2) [above left=3mm and 10mm of b3] {};
  \node (i1) [above right= 8mm and 10mm of c2] {};
  \node (i2) [below left= 3mm and 10mm of a3] {};

  \path[->]
  (f1) edge node[above] {$-$} node[below] {$f$} (f2)
  (g1) edge node[above] {$g$} node[below] {$+$} (g2)
  (i1) edge node[above] {$-$} node[below] {$i$} (i2)
  (j1) edge node[above] {$j$} node[below] {$+$} (j2);

\end{tikzpicture}

    \caption{A d-pushout of (Hasse diagrams of) c-posets. The red edges are added in
      the construction of the precedence for the pushout c-poset (the dashed red edge is subsumed by transitivity).}
\label{f:cposet-concat}
\end{figure}

\begin{remark}
d-pushouts in $\CSet$ do not coincide with d-pushouts in $\cPos$. It is apparent from \eqref{e:concat-precedence} that d-pushouts of c-sets in $\cPos$, considered as c-posets with discrete precedence, need not be c-sets because of the precedence introduced in the d-pushout.  
\end{remark}

In the remainder of the section, we show several facts about pushouts. These are needed below to show that $\cPos$ is a d-cocompletion of $\cSet$.

\begin{lemma}
\label{l:finite-cpos-generation}
    Let $p$ be a minimal and $p'$ a maximal element of a c-poset $P$ such that $p<p'$. Then 
  \begin{equation*}
    \begin{tikzcd}
      P\setminus\{p,p'\}\ar[r,"g","+" swap]\ar[d,"f" swap, "-"]
      & P\setminus\{p\}\ar[d,"j","-" swap] \\
      P\setminus\{p'\}\ar[r,"i" swap, "+"] & P
    \end{tikzcd}
  \end{equation*}
  is a d-pushout in $\cPos$, where the subsets of $P$ are regarded as c-posets with restricted precedence and the c-poset maps are the unique for- and backmorphisms induced by the inclusions.
\end{lemma}
\begin{proof}
It is immediate to verify that this square is an instance of \eqref{e:cposet-concatenation}.
\end{proof}  

\begin{proposition}
\label{p:ConcatPullback}
Every d-pushout 
	in $\cPos$ is a d-pullback.
\end{proposition}
\begin{proof}
  By Theorem \ref{t:cpos-has-proper-pushouts}, every d-pushout in
  $\cPos$ is a d-pushout in $\cSet$ after applying the forgetful
  functor $\Forget\colon \cPos\to \CSet$, and a pullback in $\CSet$ by
  Lemma \ref{l:dpushouts-are-pullbacks}. Consider a d-pushout diagram
  as in \eqref{e:cposet-concatenation}. For every commuting square in
  $\cPos$, formed by $P\xleftarrow{x}M\xrightarrow{y} Q$ and
  $P\forrightarrow{i}S\backleftarrow{j} Q$, there is a unique c-map
  $w\colon M\to R$ satisfying $fw=x$ and $gw=y$.  It remains to show
  that $w$ is a c-poset map.  It reflects precedence since $f$
  preserves and $x$ reflects precendence.  For $r,r'\in R$ we have
\begin{equation*}
  r<_R r'
  \iff
  f(r)<_P f(r')
  \implies
  \varepsilon_x(f(r))\leq_\Omega \varepsilon_x(f(r'))
  \iff
  \varepsilon_w(r)\leq_\Omega \varepsilon_w(r'),
\end{equation*}
which shows that $\varepsilon_w$ is weakly order preserving.
\end{proof}

\begin{lemma}
\label{l:ConcatFactor}
Let $S=\colim\big(P\backleftarrow{f}R \forrightarrow{g} Q\big)$ be a d-pushout in $\cPos$.
Then every c-poset map $t\colon U\to S$ from a c-set $U$ factors through $P$ or $Q$.
\end{lemma}
\begin{proof}
As usual we assume a d-pushout square as in \eqref{e:cposet-concatenation}.
Suppose there are $s\in i(P)\setminus j(Q)$ and $s'\in j(Q)\setminus i(P)$ such that $\varepsilon_t(s)\neq 1$ and $\varepsilon_t(s')\neq 0$. By \eqref{e:concat-precedence} we have $s<s'$, which enforces $\varepsilon_f(s)=\varepsilon_f(s')=*$. Thus $s,s'\in t(U)$ and $t^{-1}(s)<t^{-1}(s')$, that is, $U$ cannot be discrete. Therefore, either $\varepsilon_f(s)=1$ for all $s\in i(P)\setminus j(Q)$, or $\varepsilon_f(s')=0$ for all $s'\in j(Q)\setminus i(P)$.
  
In the first case, $f(U)\subseteq j(Q)$, so there exists a function  $k\colon U\to Q$ such that $j\circ k=t$ in $\Set$. Defining $\varepsilon_k=\varepsilon_t\circ j$ makes $k$ into a c-map and in fact a c-poset map: the activity map $\varepsilon_k$ is weakly order-preserving because both $\varepsilon_t$ and $j$ have this property (recall that $j$ is a backmorphism and hence order-preserving). Furthermore, for $s\in S$ we have
    \[
        \varepsilon_{jk}(s)=
        \begin{cases}
            \varepsilon_j(s) & \text{for }s\in i(P)\setminus j(Q),\\
            \varepsilon_k(j^{-1}(s)) & \text{for }s\in j(Q)
        \end{cases}
        =
        \begin{cases}
            1 & \text{for } s\in i(P)\setminus j(Q),\\
            \varepsilon_t(s) & \text{for } s\in j(Q)
        \end{cases}
        =
        \varepsilon_t(s),
    \]
    so that $t=jk$ in $\cSet$.
    The second case is similar.
\end{proof}

The final proposition in this section is a consequence of Proposition~\ref{p:ConcatPullback} and Lemma \ref{l:ConcatFactor}.

\begin{proposition}
  \label{p:ColimitPPreservation}
  The restricted Yoneda functor $\yoneda^\cPos_\CSet$ preserves d-pushouts.
\end{proposition}
\begin{proof}
  Consider a d-pushout of the form \eqref{e:cposet-concatenation}
in $\cPos$. Colimits of presheaves are computed
pointwise, so we need to show that the diagram
 \begin{equation*}
 \begin{tikzcd}
   R^\wedge_\CSet[U]
   \ar[r,"{g^\wedge_\CSet[U]}"]
   \ar[d,"{{f^\wedge_\CSet}[U]}",swap]
   &	{Q^\wedge_\CSet}[U]
   \ar[d,"{{j^\wedge_\CSet}[U]}"]
   \\
   {P^\wedge_\CSet}[U]
   \ar[r,"{{i^\wedge_\CSet}[U]}", swap]
   &	{S^\wedge_\CSet}[U]
 \end{tikzcd}
 \quad
 =
 \quad
 \begin{tikzcd}
   \cPos(U,R)
   \ar[r,"{g\circ(-)}"]
   \ar[d,"{f\circ(-)}",swap]
   &
   \cPos(U,Q)
   \ar[d,"{j\circ(-)}"]
   \\
   \cPos(U,P)
   \ar[r,"{i\circ(-)}", swap]
   &	
   \cPos(U,S)
 \end{tikzcd}
\end{equation*}
is a pushout in $\Set$ for every $U\in\Ob(\CSet)$. The diagram on the
right is a pullback  by Proposition \ref{p:ConcatPullback}, and all
functions in the diagram are injective as c-maps are monic by
Lemma~\ref{l:cset-properties}.\ref{i:monic}. Finally, we have that
  $\im(i^\wedge_\CSet)\cup \im(j^\wedge_\CSet)=S^\wedge_\CSet[U]$
by Lemma \ref{l:ConcatFactor}, so that the diagram is a pushout.
\end{proof}

\section{Directed cocompletions}
\label{s:d-cocompletions}

The category $\CSet$ has all d-pushouts, but they are not universal -- the Yoneda functor $\yoneda_\CSet$ does not preserve them. In this section, we show how any d-category $\cC$ can be extended to a d-category $\cP$ with universal d-pushouts which ideally is smaller than $\psh\cC$. We call such an extension a directed cocompletion, or d-cocompletion for short. We prove that the inclusion $\CSet\subseteq\cPos$ is an example of a d-cocompletion.
Then we show how a d-cocompletion of a given d-category produces d-cocompletions of its comma categories and subcategories.

\begin{definition} 
\label{d:d-pushout-generation}
    A d-category $\cP$ is \emph{d-generated} by a class $\mathscr{A}\subseteq \Ob(\cP)$ if $\Ob(\cP)=\mathscr{P}$ for every class $\mathscr{A}\subseteq\mathscr{P}\subseteq \Ob(\cP)$ that is closed under d-pushouts.
\end{definition}

\begin{definition} 
\label{d:dCocompletion}
    A d-category $\cP$ is a \emph{d-cocompletion} (\emph{by d-pushouts}) of its d-subcategory $\cC$ if
    \begin{enumerate}
        \item  $\cC$ is a full d-subcategory of $\cP$ (that is, $\cC\hookrightarrow \cP$),
        \item $\cP$ is d-cocomplete,
        \item $\cP$ is d-generated by $\Ob(\cC)$,
        \item the restricted Yoneda functor $\yoneda^\cP_\cC\colon \cP\to\psh\cC$ sends d-pushouts to pushouts.
    \end{enumerate}
    We also refer to the inclusion $\cC\hookrightarrow \cP$ as a d-cocompletion.
\end{definition}
While the above definition generalises to arbitrary fully faithful d-functors, we prefer the version given for its simplicity.

\begin{lemma}
\label{l:cpos-generated-by-xi}
The d-category $\cPos$ is d-generated by $\Ob(\CSet)$.
\end{lemma}
\begin{proof} 
Let $\Ob(\CSet)\subseteq \mathscr{P}\subseteq  \Ob(\cPos)$ be a class closed under d-pushouts. We show that every c-poset $P$ belongs to $\mathscr{P}$ by induction on the size of $P$. Clearly $\emptyset\in\mathscr{P}$. Suppose $P\neq\emptyset$ and all c-posets of size less than $|P|$ are in $\mathscr{P}$. If $P$ is discrete, then $P\in \Ob(\CSet)\subseteq \mathscr{P}$. Otherwise there are $p,p'\in P$ such that $p$ is minimal, $p'$ is maximal and $p<p'$. By Lemma \ref{l:finite-cpos-generation}, $P$ is a d-pushout of smaller c-posets which are in $\mathscr{P}$ by the induction hypothesis. Hence $P\in\mathscr{P}$.
\end{proof}

\begin{proposition}
\label{p:concset-cocompletion}
The category $\cPos$ is a d-cocompletion of $\CSet$.
\end{proposition}
\begin{proof}
Conditions (1)--(4) of Definition~\ref{d:dCocompletion} have been checked in Lemma \ref{l:xi-cpos-functors}, Theorem \ref{t:cpos-has-proper-pushouts}, Lemma \ref{l:cpos-generated-by-xi} and Proposition \ref{p:ColimitPPreservation}, respectively.
\end{proof}

Next, we present a crucial property of d-cocompletions.

\begin{proposition}
\label{p:dCocoImpliesFaithful}
If $\cC\hookrightarrow\cP$ is a d-cocompletion, then
  $\yoneda^\cP_\cC$ is fully faithful. 
\end{proposition}
\begin{proof}
	Let $\mathscr{P}\subseteq \Ob(\cP)$ be a class of objects $P$ for which
	$
		\yoneda^\cP_\cC\colon\cP(P,Q)\to \psh\cC(P^\wedge_\cC,Q^\wedge_\cC)
	$
	is bijective for all $Q\in\Ob(\cP)$.
   If $U\in \Ob(\cC)$, then 
   \[
    \cP(U,Q)
    \cong
    Q^\wedge_{\cC}[U]
    \cong
    \psh\cC(U^\wedge_\cC,Q^\wedge_\cC)
  \]
  and thus $\Ob(\cC)\subseteq \mathscr{P}$.
  Next, suppose $P,Q,R\in\mathscr{P}$ and
  $
    S=\colim\big(P\backleftarrow{f} R \forrightarrow{g} Q\big).
  $
  Then, for all  $M\in\cP$,
  \begin{align*}
    \cP(S,M)
    &\cong
    \lim\Big(
      \cP(P,M)\xrightarrow{(-)\circ f}\cP(R,M)\xleftarrow{(-)\circ g}\cP(Q,M)
    \Big)
    \\
    &\cong
    \lim\Big(
      \psh\cC(P^\wedge_\cC,M^\wedge_\cC)
	      \xrightarrow{(-)\circ f^\wedge_\cC}
      \psh\cC(R^\wedge_\cC,M^\wedge_\cC)
      	\xleftarrow{(-)\circ g^\wedge_\cC}
      \psh\cC(Q^\wedge_\cC,M^\wedge_\cC)
    \Big)
    \\
    &
\cong
    \psh\cC\Big(
      \colim\big(P^\wedge_\cC\xleftarrow{f^\wedge_\cC}R^\wedge_\cC \xrightarrow{g^\wedge_\cC}Q^\wedge_\cC\big),
      M^\wedge_\cC
    \Big)
    \\
    &
    \cong
    \psh\cC(S^\wedge_\cC, M^\wedge_\cC).
  \end{align*}
   The first and third step use the universal property of pushouts. The second follows from the assumption that $P,Q,R\in\mathscr{P}$ and
  the fourth from \ref{d:dCocompletion}(4). Thus $S\in\mathscr{P}$. Finally, \ref{d:dCocompletion}(3) implies that $\mathscr{P}=\Ob(\cP)$.
\end{proof}

We now describe a situation where d-cocompletions extend to comma categories. We start with two technical facts.

\begin{proposition}
\label{p:slice-creates-pushouts}
Let $\cP$ be a d-category, $\cE$ a category, $F\colon \cP\to\cE$ a functor that sends d-pushouts to pushouts and $L\in \Ob(\cE)$. Then the projection functor $\ProjFun\colon\cP\slice_FL\to \cP$ creates d-pushouts. (Recall that $\cP\slice_FL$ has d-structure reflected by $\ProjFun$.)
\end{proposition}

\begin{proof}
Consider the d-span $\big((P,\lambda)\backleftarrow{f}(R,\nu)\forrightarrow{g}(Q,\mu)\big)$ in $\cP\slice_F L$ in the diagram on the left in Figure~\ref{f:slice-pushout}. 

\begin{figure}
    \centering
   \begin{tikzcd}[
    row sep=.7cm, 
    column sep=.2cm, 
    cells={nodes={minimum width=1.2cm, minimum height=0.8cm, inner sep=0pt}
    }
]
    &
    (R,\nu)\ar[dl, "-", "f" swap]\ar[dr, "g", "+" swap] &&&&
    R\ar[dl, "-", "f" swap]\ar[dr, "g", "+" swap]  &&&&
    F(R)\ar[dl,"F(f)"' {xshift=4pt}]\ar[dr,"F(g)" {xshift=-4pt}]\ar[ddrrrr, bend left=40, "\nu", color=blue]
    \\
    (P,\lambda)\ar[dr, dashed, "+", "i" swap]\ar[ddr, bend right=30, dotted, swap, "x"] &&
    (Q,\mu)\ar[dl, dashed, "j", "-" swap]\ar[ddl, dotted, bend left=30, "y"]&&
    P\ar[dr, "+", "i" swap, dashed]\ar[ddr, dotted, bend right=30, swap, "x"] &&
    Q\ar[dl, "j", "-" swap, dashed]\ar[ddl, dotted, bend left=30, "y"] &&
    F(P)\ar[dr,"F(i)"' {xshift=4pt}, dashed]\ar[ddr, bend right=30, "F(x)"' {xshift=4pt},dotted]  &&
    F(Q)\ar[dl,"F(j)" {xshift=-4pt}, dashed]\ar[ddl, bend left=30, "F(y)" {xshift=-4pt}, pos=0.65, dotted] \ar[drrr, bend left=10, "\mu", color=blue]
    \\
    &
    (S,\xi)\ar[d, dotted, swap, "z"]&&&&
    S\ar[d, dotted, swap, "z"] &&&&
    F(S)\ar[d, "F(z)"' {xshift=2pt, yshift=2pt}, dotted] \ar[rrrr, "\xi", pos=0.2, bend right=10, dashed, color=blue]&&&&
    L \ar[from=ulllll, bend left=10, "\lambda", pos=0.3, color=blue]
    \\
    &
    (T,\tau) &&&&
    T &&&&
    F(T)\ar[urrrr, bend right=20, "\tau"', color=blue, dotted]
    \\
    &
    \cP\slice_F L\ar[rrrr, "\ProjFun"] &&&&
    \cP\ar[rrrr, "F"]&&&&
    \cE
\end{tikzcd}
    \caption{d-pushout diagrams for the d-pushout construction in the comma category in Proposition~\ref{p:slice-creates-pushouts}.}
    \label{f:slice-pushout}
\end{figure}
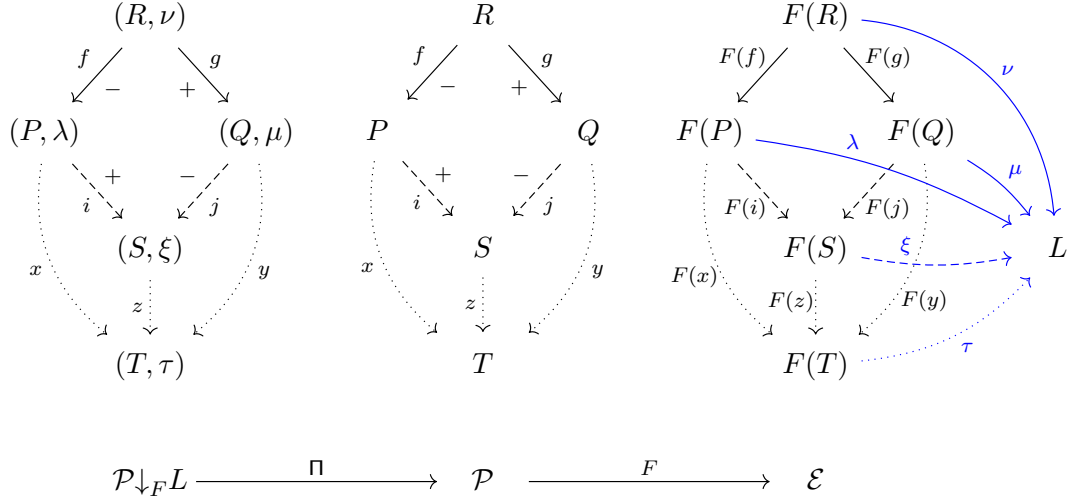

For d-pushout creation, consider further a d-pushout  $S=\colim\big(P\backleftarrow{f} R\forrightarrow{g} Q\big)$ in $\cP$, as in the square in the central diagram of Figure~\ref{f:slice-pushout}. Its image under $F$, in the diagram in $\cE$ on the right in this Figure, is a pushout by assumption. As the morphisms $\lambda$, $\mu$ and $\nu$ form a cocone, there is a unique morphism $\xi\colon F(S)\to L$ in $\cE$ such that 
\begin{equation}
\label{e:xi-unique}
  \xi\circ F(i)= \lambda\qquad\text{and}\qquad \xi\circ F(j)=\mu. 
\end{equation} 
Thus $i\colon (P,\lambda)\to (S,\xi)$ and $j\colon (Q,\mu)\to(S,\xi)$ are morphisms in $\cP\slice_FL$, completing the d-square in the left diagram, which commutes because it commutes in $\cP$.
It remains to show that this d-square is a pushout, that is, $(S,\xi) = \colim\big((P,\lambda)\backleftarrow{f}(R,\nu)\forrightarrow{g}(Q,\mu)\big)$ in $\cP\slice_F L$.

Assume that the morphisms $x\colon (P,\lambda)\to (T,\tau)$ and $y\colon (Q,\mu)\to (T,\tau)$ in $\cP\slice_F L$, for some $(T,\tau)\in \Ob(\cP\slice_F L)$ with $\tau\colon F(T)\to L$, form another commuting square, as indicated above. As morphisms in $\cP\slice_F L$, they satisfy 
\begin{equation}
\label{e:xy-comma}
  \lambda = \tau\circ F(x)\qquad\text{and}\qquad \mu = \tau\circ F(y).  
\end{equation}
 Since $x\colon P\to T$ and $y\colon Q\to T$ form a cocone for the d-cospan $(P\backleftarrow{f}R\forrightarrow{g}Q)$, there is thus a unique $z\colon S\to T$ in $\cP$ for which $x=zi$ and $y=zj$. Combining \eqref{e:xi-unique} and \eqref{e:xy-comma} with these identities, $\xi\circ F(i)=\tau\circ F(z)\circ F(i)$ and  $\xi\circ F(j)=\tau\circ F(z)\circ F(j)$, so that $\tau\circ F(z)=\xi$, since the solid square at the right is a pushout.
 Thus $z\colon (S,\xi)\to (T,\tau)$ is indeed a morphism in $\cP\slice_F L$ and in fact the unique morphism for which the diagram in $\cP\slice_F L$ commutes.
\end{proof}

\begin{lemma}
\label{l:slice-pushout-reflection}
    Let $F\colon\cC\to\cE$ be a functor and $L\in\Ob(\cE)$.
    Then the pushforward functor 
    \[
        \ProjFun_*\colon\psh{\cC\slice_FL}\to \psh\cC
    \]
    reflects pushouts.
\end{lemma}
\begin{proof}
    Let $D\colon\lozenge\to\psh{\cC\slice_FL}$ be a diagram such that $\ProjFun_*\circ D$ is a pushout. Then the diagram $(\ProjFun_*\circ D)[U]$ is a pushout for every $U\in\Ob(\cC)$. For every presheaf $X\in\psh{\cC\slice_FL}$ there is the decomposition 
    \[
        \ProjFun_*(X)[U]\cong \coprod_{\lambda\in \cE(F(U),L)} X[(U,\lambda)].
    \]
    Thus $(\ProjFun_*\circ D)[U]\cong \coprod_{\lambda\in \cE(F(U),L)} D[(U,\lambda)]$. A coproduct of pushout diagrams in $\Set$ is a pushout diagram if and only if  all summands are pushouts. Thus $D[(U,\lambda)]$ is a pushout for all $(U,\lambda)\in\Ob(\cC\slice_F L)$, which implies that $D$ is a pushout.
\end{proof}

\begin{theorem}
\label{t:SlicedCocompletion}
    Let $\cC\hookrightarrow \cP$ be a d-cocompletion, $F\colon \cP\to\cE$ a functor sending d-pushouts to pushouts, and $L\in\Ob(\cE)$.
    Then the inclusion $\cC\slice_{F|_{\cC}} L \hookrightarrow \cP\slice_F L$ is a d-cocompletion.
\end{theorem}
\begin{proof}
We need to check the conditions of Definition \ref{d:dCocompletion}. Condition (1) follows immediately and (2) follows from d-cocompleteness of $\cP$ by Proposition \ref{p:slice-creates-pushouts}.
We write $\cC\slice L$ and $\cP\slice L$ for $\cC\slice_{F|_{\cC}}L$ and $\cP\slice_FL$, respectively.
    
For condition (3), let $\mathscr{Q}$ be a subclass of $\Ob(\cP\slice L)$ that contains $\Ob(\cC\slice L)$ and is closed under d-pushouts.  The class 
$
\mathscr{P}=\{P\in\Ob(\cP)\mid \forall\lambda\colon F(P)\to L.\ (P,\lambda)\in\mathscr{Q}\}
$
contains $\Ob(\cC)$ by definition. We show that it is closed under d-pushouts. So consider the d-pushout  in $\cP$ in the following diagram on the left (taken from Figure~\ref{f:slice-pushout}).  
\begin{equation*}
\begin{tikzcd}
        R\ar[r,"g","+" swap]\ar[d,"f" swap,"-"]
        &
        Q\ar[d,"j", "-" swap]
    \\
        P\ar[r,"i" swap, "+"]
        &
        S
\end{tikzcd}
\qquad\qquad\qquad
\begin{tikzcd}
        (R,\nu)\ar[r,"g","+" swap]\ar[d,"f" swap,"-"]
        &
        (Q,\mu)\ar[d,"j", "-" swap]
    \\
        (P,\lambda)\ar[r,"i" swap, "+"]
        &
        (S,\xi)
\end{tikzcd}
\end{equation*}
Fixing the morphism $\xi\colon F(S)\to L$ and defining
\begin{equation*}
  \lambda=\xi\circ F(i),\qquad \mu=\xi\circ F(j),\qquad \nu=\xi\circ F(if)=
  \xi\circ F(jg)
\end{equation*}
yields the commutative square
in $\cP\slice_F L$ on the right (taken again from Figure~\ref{f:slice-pushout}), which is a d-pushout by Proposition \ref{p:slice-creates-pushouts}. Then $(S,\xi)\in \mathscr{Q}$ whenever $(P,\lambda),(Q,\mu),(R,\nu)\in\mathscr{Q}$. This argument holds for arbitrary $\xi$, so that $P\in\mathscr{P}$ whenever $P,Q,R\in\mathscr{P}$. As $\cP$ is a d-cocompletion of $\cC$, $\mathscr{P}=\Ob(\cP)$ and $\mathscr{Q}=\Ob(\cP\slice L)$. Thus $\cP\slice L$ is generated from $\cC\slice L$ by d-pushouts.

For condition (4), suppose that $D\colon \lozenge\to \cP\slice L$ is a d-pushout. In the commutative diagram 
\[
    \begin{tikzcd}
        \lozenge\ar[r,"D"]& \cP\slice L
            \ar[r,"\yoneda^{\cP\slice L}_{\cC\slice L}"]
            \ar[d,"\ProjFun" swap] 
        & \psh{\cC\slice L} \ar[d,"\ProjFun_*"]
        \\
        & \cP \ar[r, swap, "\yoneda^\cP_\cC"]
        & \psh{\cC}
    \end{tikzcd}
\]
the projection d-functor $\ProjFun$ preserves d-pushouts since it creates d-pushouts by Proposition~\ref{p:slice-creates-pushouts}. Moreover, $\yoneda_{\cC}^\cP$ sends d-pushouts to pushouts because $\cC\hookrightarrow \cP$ is a d-cocompletion and $\ProjFun_*$ reflects pushouts by Lemma \ref{l:slice-pushout-reflection}. It follows that $\yoneda^\cP_\cC\circ \ProjFun\circ D=\ProjFun_*\circ \yoneda^{\cP\slice L}_{\cC\slice L}\circ D$ is a pushout and so is $\yoneda^{\cP\slice L}_{\cC\slice L}\circ D$. 
\end{proof}

In the final part of this section, we apply Theorem \ref{t:SlicedCocompletion} to a special case to show that d-cocomple\-tions can be restricted in domain and codomain to d-subcategories with sieve-like properties.

\begin{definition}
\label{d:subobject-closed}
A subcategory $\cD$ of $\cC$ is a \emph{sieve}~\cite[Def 4.1]{ArtinGV72} if it is full and, for all $U,V \in \Ob(\cC)$,
\[
    \forall f\in \cC(U,V).\ V\in\Ob(\cD)\implies U\in \Ob(\cD).
\]
\end{definition}

The next lemma follows immediately from the definition.
\begin{lemma}
\label{l:indicator-functor}
    A full subcategory $\cD$ of $\cC$ is a sieve if and only if there is a functor $\chi^\cC_\cD\colon\cC\to\{1\to 0\}$ such that $P\in\Ob(\cD)\iff \chi^\cC_\cD(U)=1$.
    Also, the functor $\cD\to\cC\slice_{\chi^\cC_\cD}1$ given by $U\mapsto(U,\id_1)$ is an isomorphism of categories.
\end{lemma}

\begin{definition}
    Let $\cP$ be a category, let $\cC$ and $\cD$ be full subcategories satisfying $\cD\subseteq \cC\subseteq \cP$. The \emph{residuation} $\cD\setminus_\cP \cC$ of $\cD$ and $\cC$ relative to $\cP$ is the full subcategory of $\cP$ with objects
    \[
        \Ob(\cD\setminus_\cP \cC)
        =
        \{P\in \Ob(\cP)\mid \forall f\in\cP(U,P).\ U\in \Ob(\cC) \implies U\in \Ob(\cD)\}.
    \]
\end{definition}
See~\cite[p.56]{MacLane-Moerdijk_Sheaves} for a related construction.

\begin{lemma}
\label{l:Q-domain-closed}
Let $\cP$ be a category, $\cC$ a full subcategory of $\cP$ and $\cD\subseteq\cC$ a sieve.
Then
\begin{enumerate}
    \item $\cD\setminus_\cP \cC\subseteq \cP$ is a sieve,
    \item $(\cD\setminus_\cP \cC)\cap \cC=\cD$,
    \item $\chi^{\cP}_{\cD\setminus_\cP \cC}|_{\cC}=\chi^\cC_\cD$ and $\cD\setminus_\cP \cC\cong \cP\slice_{\chi^{\cP}_{\cD\setminus_\cP \cC}}1$.
\end{enumerate}
\end{lemma}
\begin{proof}
    We abbreviate $\cQ=\cD\setminus_\cP \cC$.
    For (1), assume by contradiction that there exist $P\notin \Ob(\cQ)$, $Q\in \Ob(\cQ)$ and $f\in \cP(P,Q)$. Since $P\notin\cQ$, there exists a morphism $g\in \cP(U,P)$ such that  $U\in\Ob(\cC)$ and $U\notin\Ob(\cD)$. But then $fg\in \cP(U,Q)$, which contradicts $Q\in\cQ$.
    
    For (2), assume that $V\in \Ob(\cC)$. Then $V\notin\Ob(\cD)$ if there exists some $f\in\cC(U,V)$ for some $U\in\Ob(\cC)\setminus\Ob(\cD)$. Moreover, $V\notin\Ob(\cQ)$ if there exists an $f'\in\cP(U,V)$ for $U$ satisfying the same condition. Since $\cC\subseteq \cP$ is a full subcategory, these two conditions are equivalent. Therefore, $\cC\cap \cQ=\cC\cap \cD=\cD$.

    Item (3) follows from (1) and (2) by Lemma \ref{l:indicator-functor}.
\end{proof}

If $\cC$ is a d-category, then a sieve $\cD$ in $\cC$ is obviously a full d-subcategory.

\begin{proposition}
\label{p:subcat-cocompletion}
    Let $\cC\hookrightarrow\cP$ be a d-cocompletion and $\cD\subseteq\cC$ a sieve. Then
    \begin{enumerate}
    \item 
        the indicator functor $\chi^\cP_{\cD\setminus_\cP\cC}:\cP\to\{1\to 0\}$ sends d-pushouts to pushouts, 
    \item 
        $\cD\hookrightarrow\cD\setminus_\cP \cC$ is a d-cocompletion.
    \end{enumerate}
\end{proposition}
\begin{proof}
    We abbreviate again $\cQ=\cD\setminus_\cP \cC$. 
    For (1), let $S=\colim \big(P\backleftarrow{}R\forrightarrow{}Q\big)$ be a d-pushout in $\cP$. Then $S^\wedge_\cC\cong \colim (P^\wedge_\cC\gets R^\wedge_\cC\to Q^\wedge_\cC)$ by Definition~\ref{d:dCocompletion}.4.
Taking the value at $U\in\Ob(\cC)$ yields
\[
        S^\wedge_\cC[U]=\cP(U,S)\cong\colim \big(\cP(U,P)\gets\cP(U,R)\to\cP(U,Q)\big).
\]
Thus $\cP(U,S)\neq \emptyset$ if and only if $\cP(U,P)\neq\emptyset$ or $\cP(U,Q)\neq\emptyset$, so that
\begin{align*}
    \chi^\cP_\cQ(S)=0
    &\iff 
    S\notin \cQ
    \\
    &\iff
    \exists U\in\Ob(\cC).\; U\notin\Ob(\cD)\;\land\; \cP(U,P)\neq\emptyset
    \\
    &\iff
    \exists U\in\Ob(\cC).\; U\notin\Ob(\cD)\;\land\; 
    (\cP(U,P)\neq\emptyset\;\lor\;\cP(U,Q)\neq\emptyset)
    \\
    &\iff 
    \chi^\cP_\cQ(P)=0\;\lor\;\chi^\cP_\cQ(Q)=0.
\end{align*}
This translates to $\chi^\cP_\cQ(S)=\colim \big(\chi^\cP_\cQ(P)\gets\chi^\cP_\cQ(R)\to\chi^\cP_\cQ(Q)\big)$ in the category $\{1\to 0\}$.

For (2), we use Theorem \ref{t:SlicedCocompletion} to obtain that
\begin{equation}
    \cD
    \cong
    \cC\slice_{\chi^\cC_\cD} 1
    \hookrightarrow
    \cP\slice_{\chi^\cP_\cQ} 1
    \cong
    \cQ,
\end{equation}
is a d-cocompletion. The assumptions of Theorem \ref{t:SlicedCocompletion} are satisfied by (1) and Lemma \ref{l:Q-domain-closed}, the equivalences follow from Lemma \ref{l:indicator-functor}.
\end{proof}

\section{Directed cospans}
\label{s:d-cospans}

Track objects in a d-category $\cC$ are, by definition, cospans of presheaves over $\cC$. We now show a d-cocompletion $\cP$ of $\cC$ can provide an explicit description of the bicategory of track objects. We introduce the notion of directed cospan and identify a class of those d-cospans (interval d-cospans) that correspond precisely to track objects. We also study  comma categories in this context.

\begin{definition}
\label{d:d-cospan}
Let $\cP$ be a d-category and let $\mathscr{A}\subseteq \Ob(\cP)$ be a class
of objects such that all d-spans 
$\big(P\backleftarrow{}U\forrightarrow{}Q\big)$ have d-pushouts for
each $U\in\mathscr{A}$.  The \emph{bicategory
  $\dCosp(\cP,\mathscr{A})$ of d-cospans} on $\cP$ restricted to
$\mathscr{A}$ is defined by the following data:
\begin{itemize}
    \item It has objects $\mathscr{A}$. 
    \item For $U,V\in\mathscr{A}$, the hom category $\dCosp(\cP,\mathscr{A})_U^V$ has all d-cospans  $
    \big(U\forrightarrow{s} P\backleftarrow{t}V\big)
    $ for $P\in \Ob(\cP)$ as objects and commutative diagrams in $\cP$
    \begin{equation}
    \label{e:dcospan-map}
    \begin{tikzcd}[row sep=0.5cm, column sep=0.8cm]
        & Q\ar[dd,"f"] & \\
        U\ar[ru,"s","+"'{inner sep=0pt}]\ar[rd,"s'"',"+"{inner sep=0pt}] & & V\ar[lu,"t"',"-"{inner sep=0pt}]\ar[ld,"t'","-"'{inner sep=0pt}] \\
        & P & 
    \end{tikzcd}
    \end{equation}
    as morphisms. 
    \item The horizontal composition of d-cospans is given by
    \[
    \big(U\forrightarrow{s} P\backleftarrow{t}V\big)
    *
    \big(V\forrightarrow{s'} Q\backleftarrow{t'}W\big)
    =
    \big(U\forrightarrow{is} R\backleftarrow{jt'}W\big),
    \]
for a d-pushout
\begin{equation}
\label{e:dcospan-composition}
    \begin{tikzcd}[row sep=0.5cm, column sep=0.8cm]
        & & R 
        & & \\
        & P\ar[ru,"i","+"'{inner sep=0pt}] & & Q \ar[lu,"j"',"-"{inner sep=0pt}] & \\
        U 
        \ar[ru,"s","+"'{inner sep=0pt}]
        & &
        V \ar[ru,"s'"',"+"{inner sep=0pt}] \ar[lu,"t","-"'{inner sep=0pt}] & &
        W 
        \ar[lu,"t'"',"-"{inner sep=0pt}]
    \end{tikzcd}
\end{equation}
\end{itemize}
\end{definition}

Suitable d-pushouts are required for constructing horizontal compositions: this guarantees that compositions of d-cospans are again d-cospans. As usual, we keep silent about coherence conditions, referring the reader to \cite{Johnson-Yau_2Categories}.

We write $\dCosp(\cP,\cC)$ for $\dCosp(\cP,\Ob(\cC))$ if $\cC$ is a subcategory of $\cP$, and $\dCosp(\cP)$ for $\dCosp(\cP,\cP)$.

\begin{remark}
    We cannot define $\dCosp(\cP)$ as a sub-bicategory of $\Cosp(\cP)$, since the latter might not even be well-defined because we do not assume that all pushouts in $\cP$ exist.
    Yet if $\cP$ is a category equipped with the full d-structure and all pushouts exist in $\cP$, then $\dCosp(\cP)=\Cosp(\cP)$.
    For similar reasons, $\dCosp(\cP,\mathscr{A})$ may be well-defined even if $\dCosp(\cP)$ is not.
\end{remark}

\begin{proposition}
\label{p:cospan-pseudofunctors}
Let $F:\cQ\to\cP$ be a d-functor, let $\mathscr{A}\subseteq \Ob(\cP)$
and $\mathscr{B}\subseteq \Ob(\cQ)$. Assume that
$F(\mathscr{B})\subseteq \mathscr{A}$ and that $F$ preserves
d-pushouts of d-spans $\big(Q\backleftarrow{}V\forrightarrow{}Q'\big)$
for all $V\in\mathscr{B}$ and all $Q,Q'\in\Ob(\cQ)$.  Then
\begin{enumerate}
    \item 
$F$ induces the pseudofunctor $F_*:\dCosp(\cQ,\mathscr{B})\to \dCosp(\cP,\mathscr{A})$
given by
\begin{align*}\
\Ob(F_*) &=F|_{\mathscr{B}}\colon \mathscr{B}\to\mathscr{A},\\
    F_*\Big(P\forrightarrow{s} R\backleftarrow{t}Q\Big)
    &=
    \Big(F(P)\forrightarrow{F(s)} F(R)\backleftarrow{F(t)}F(Q)\Big),
\end{align*}
\item $F_*$ is locally fully faithful if $F$ is fully faithful.
\end{enumerate}
\end{proposition}
\begin{proof}
A similar fact for spans is proven in~\cite[Proposition 4.1.24]{Johnson-Yau_2Categories}; a dual argument applies here.
\end{proof}

\begin{proposition}
\label{p:cospan-pseudofunctors2}
Let $F:\cQ\to\cP$ be a functor from a d-category $\cQ$ to a category $\cP$ and let $\mathscr{A}\subseteq \Ob(\cC)$, $\mathscr{B}\subseteq \Ob(\cD)$.
Assume that $F(\mathscr{B})\subseteq \mathscr{A}$ and that $F$ sends d-pushouts of d-spans $\big(Q\backleftarrow{}V\forrightarrow{}Q'\big)$ to pushouts, for $V\in\mathscr{B}$ and $Q,Q'\in\Ob(\cQ)$.  Then
\begin{enumerate}
    \item $F$ induces a pseudofunctor $F_*:\dCosp(\cQ,\mathscr{B})\to \Cosp(\cP,\mathscr{A})$ defined as in Proposition~\ref{p:cospan-pseudofunctors};
    \item $F_\ast$ is locally fully faithful whenever $F$ is fully faithful.
\end{enumerate}
\end{proposition}
\begin{proof}
    We equip $\cP$ with the full d-structure. The conclusion then follows from Proposition~\ref{p:cospan-pseudofunctors}.
\end{proof}

Constructions of track objects and interval cospans in the next two subsections require the following two induced pseudofunctors.
\begin{lemma}
\label{l:examples-of-induced-pseudofunctors}
    Let $\cC\hookrightarrow\cP$ be a d-cocompletion.
    \begin{enumerate}
    \item The restricted Yoneda functor $\yoneda^\cP_\cC$ induces the locally fully faithful pseudofunctor
    \[
    (\yoneda^\cP_\cC)_*:
    \dCosp(\cP,\cC)
    \to
    \Cosp(\psh\cC,\cC).
    \]
    \item 
    If the functor $F:\cP\to\cE$ sends d-pushouts to pushouts and $L\in\Ob(\cE)$, then the projection functor $\ProjFun:\cP\slice_F L\to\cP$ induces the pseudofunctor
    \[
    \ProjFun_*:\dCosp(\cP\slice_FL,\cC\slice_{F|_{\cC}}L)\to \dCosp(\cP,\cC).
    \]
    \end{enumerate}
\end{lemma}
\begin{proof}
    The d-cospan bicategories considered in (1) and (2) are well-defined: both $\cP$ and $\cP\slice_FL$ are d-cocomplete by condition \ref{d:dCocompletion}(2) and Theorem \ref{t:SlicedCocompletion}, respectively. Hence (1) and (2) follow from Propositions \ref{p:cospan-pseudofunctors2} and \ref{p:cospan-pseudofunctors} applied to $\yoneda^\cP_\cC$ and $\ProjFun$, respectively. The functors send d-pushouts to (d-)pushouts by condition \ref{d:dCocompletion}(4) and Proposition \ref{p:slice-creates-pushouts} while $\yoneda^\cP_\cC$ is fully faithful by Proposition \ref{p:dCocoImpliesFaithful}.
\end{proof}

Within the bicategory of d-cospans, we now define the locally full sub-bicategory of interval d-cospans and relate it with the bicategory of track objects.
Recall from Section~\ref{s:preliminaries} that the bicategory of track objects $\tro\cC$ is the smallest locally strictly full sub-bicategory of $\Cosp(\psh\cC,\cC)$ that contains all elementary track objects $\big(V^\wedge_\cC\xrightarrow{f^\wedge_\cC} U^\wedge_\cC \xleftarrow{g^\wedge_\cC} W^\wedge_\cC\big)$ for all $U,V,W\in\Ob(\cC)$, $f\in\cC^+(V,U)$ and $g\in \cC^-(W,U)$.

\begin{definition}
\label{d:interval-d-cospan}
    Let $\cC\hookrightarrow \cP$ be a d-cocompletion (and hence $\cP$ d-cocomplete).
\begin{enumerate}
\item
    A d-cospan in $\dCosp(\cP,\cC)$ is \emph{elementary} if it is of the form $(V\forrightarrow{} U \backleftarrow{} W)$ for some $U,V,W\in\Ob(\cC)$ (as an object of $\dCosp(\cP,\cC)_V^W$
    ).
\item
    The \emph{interval d-cospan bicategory} $\dICosp(\cP,\cC)$ is the smallest locally strictly full sub-bicategory of $\dCosp(\cP,\cC)$ that contains all elementary d-cospans.
\end{enumerate}
We refer to d-cospans in $\dICosp(\cP,\cC)$ as  \emph{interval d-cospans}.
\end{definition}

The motivation for the term ``interval'' becomes clear in Section \ref{s:enriched-posets}, where we establish a relationship between interval d-cospans in $\dICosp(\cPos,\cSet)$ and interval posets.

The following theorem establishes a relationship between track objects and interval d-cospans.

\begin{theorem}
\label{t:dcospans-are-track-objects}
    For each d-cocompletion $\cC\hookrightarrow\cP$ there exists a unique pseudofunctor $\Phi^\cP_\cC$ for which the following diagram of bicategories and pseudofunctors commutes: 
    \[
    \begin{tikzcd}
        \dICosp(\cP,\cC)
            \ar[d,dashed,"\Phi_\cC^\cP"']
            \ar[rr,hook]
    &&
        \dCosp(\cP,\cC)
            \ar[d,"(\yoneda^\cP_\cC)_*"]
    \\
        \tro\cC
            \ar[r,hook]
    &
        \saut\cC
            \ar[r,double distance=2pt,-]
    &
        \Cosp(\psh\cC,\cC)
    \end{tikzcd}
    \]
    Moreover, $\Phi^\cP_\cC$ is a biequivalence.
\end{theorem}
\begin{proof}
    The class of objects of all bicategories in the above diagram is $\Ob(\cC)$;  all solid arrows are identities on objects.
    Further, the pseudofunctor $(\yoneda^\cP_\cC)_*$ maps
    elementary d-cospans in $\dICosp(\cP,\cC)$ to elementary track objects in $\Cosp(\psh\cC,\cC)$. Since every d-cospan in $\dICosp(\cP,\cC)$ is isomorphic to a horizontal composition of elementary d-cospans, its image under  $(\yoneda^\cP_\cC)_*$ is a track object. The restriction of $(\yoneda^\cP_\cC)_*$ to $\dICosp(\cP,\cC)$ therefore factors through $\tro\cC$.
    Moreover, all solid pseudofunctors in the diagram are locally fully faithful: by Lemma \ref{l:examples-of-induced-pseudofunctors}(1) for $(\yoneda^\cP_\cC)_*$ and by definition for the two horizontal inclusions. 
\end{proof}

Finally, we show that a d-cospan in a comma category is interval if and only if projects to an interval d-cospan in the base category.

\begin{proposition}
\label{p:interval-slice}
Let $\cC\hookrightarrow\cP$ be a d-cocompletion, $F\colon\cP\to\cE$  a functor sending d-pushouts to pushouts and $L\in\Ob(\cE)$.
Consider the following commutative diagram of bicategories (we write $\cP\slice L$ for $\cP\slice_F L$ and $\cC\slice L$ for $\cC\slice_{F|_\cC} L$):
\begin{equation}
\label{e:slice-interval}
\begin{tikzcd}
    \dICosp(\cP\slice L,\cC\slice L)
        \ar[r,hook]
        \ar[d,"\ProjFun_*"']
    &
    \dCosp(\cP\slice L,\cC\slice L)    
        \ar[d,"\ProjFun_*"]
\\
    \dICosp(\cP,\cC)
        \ar[r,hook]
    &
    \dCosp(\cP,\cC)    
\end{tikzcd}
\end{equation}
where $\ProjFun_*$ is the pseudofunctor induced by $\ProjFun:\cP\slice_FL\to\cP$ from Lemma \ref{l:examples-of-induced-pseudofunctors}. Then
    \[
        \dICosp(\cP\slice L,\cC\slice L)
        =
        (\ProjFun_*)^{-1}(\dICosp(\cP,\cC)).
    \]  
\end{proposition}
\begin{proof}
Since $\ProjFun$ is a d-functor, the pseudofunctor $\ProjFun_*$ sends elementary d-cospans to elementary d-cospans, so that
\[
        \ProjFun_*(\dICosp(\cP\slice L,\cC\slice L))
        \subseteq
        \dICosp(\cP,\cC).
\] 
For the converse inclusion, consider the class 
\[
    \mathscr{P}=\left\{\big(U\forrightarrow{f}P\backleftarrow{g}V\big)\in \dCosp(\cP,\cC)\mid \ProjFun^{-1}(U\forrightarrow{f}P\backleftarrow{g}V)\subseteq \dICosp(\cP\slice L, \cC\slice{L})\right\}.
\]
It contains all elementary d-cospans.
Consider a horizontal composition of d-cospans
\[
    \big(P\forrightarrow{s}S\backleftarrow{t}Q\big)
    *
    \big(Q\forrightarrow{s'}T\backleftarrow{t'}R\big)
    =
    \big(P\forrightarrow{is}X\backleftarrow{jt'}R\big)
\]
as in \eqref{e:dcospan-composition} and suppose $\big(P\forrightarrow{s}S\backleftarrow{t}Q\big), \big(Q\forrightarrow{s'}T\backleftarrow{t'}R\big)\in\mathscr{P}$. For every $\xi\in \cE(F(X),L)$ we have 
\begin{equation*}
    (X,\xi)=\colim\big((S,\xi\circ F(i))\backleftarrow{t} (Q,\xi\circ F(it))\forrightarrow{s'}(T,\xi\circ F(j))\big)
\end{equation*} since $\ProjFun$ creates d-pushouts by Proposition \ref{p:slice-creates-pushouts}. Thus each d-cospan in $(\ProjFun_*)^{-1}\big(P\forrightarrow{is}X\backleftarrow{jt'}R\big)$ is a composition of d-cospans in $\dICosp(\cP\slice L,\cC\slice L)$.
It follows that $\mathscr{P}$ is closed under horizontal composition and therefore
$
    \ProjFun_*^{-1}(\dICosp(\cP,\cC))\subseteq \ProjFun_*^{-1}(\mathscr{P}) \subseteq \dICosp(\cP\slice{L},\cC\slice{L})
$.
\end{proof}

As sieves are a special case of comma categories, we obtain the following corollary.

\begin{lemma}
\label{l:interval-sieve}
    Let $\cC\hookrightarrow\cP$ is a d-cocompletion and $\cD\subseteq\cC$ a sieve. Then
    \[
    \dICosp(\cD\setminus_\cP \cC,\cD)
    \quad\text{and}\quad
    \dICosp(\cC,\cP)\cap \dCosp(\cD\setminus_\cP \cC,\cD)
    \]
    are identical sub-bicategories of $\dCosp(\cC,\cP)$.
\end{lemma}
\begin{proof}
    By Proposition \ref{p:subcat-cocompletion}, the indicator functor $\chi^\cP_{\cD\setminus_\cP\cC}$ sends pushouts to pushouts and hence satisfies the conditions of Proposition \ref{p:interval-slice}. Taking into account that $\chi^\cP_{\cD\setminus_\cP\cC}|_{\cC}=\chi^\cC_\cD$, $\cD\setminus_\cP\cC\cong \cP\slice_{\chi^\cP_{\cD\setminus_\cP\cC}}1$ and $\cD\cong \cC\slice_{\chi^\cC_\cD}1$ by Lemmas \ref{l:indicator-functor} and \ref{l:Q-domain-closed}, Diagram \eqref{e:slice-interval} takes the form 
    \[
\begin{tikzcd}
    \dICosp(\cD\setminus_\cP\cC,\cD)
        \ar[r,hook]
        \ar[d,hook]
    &
    \dCosp(\cD\setminus_\cP\cC,\cD)    
        \ar[d,hook]
\\
    \dICosp(\cP,\cC)
        \ar[r,hook]
    &
    \dCosp(\cP,\cC)    
\end{tikzcd}    
    \]
    The conclusion follows from Proposition \ref{p:interval-slice}.
\end{proof}

\section{Enriched c-sets and c-posets} 
\label{s:enriched-posets}
\def\PE{D}

The pomsets considered in \cite{FJSZ_HDALang, FJSZ-Kleene} are not merely c-posets. Beyond precedences, they are equipped with an event order and a labelling. We therefore generalise Proposition~\ref{p:concset-cocompletion}  to a larger class of categories beyond $\cSet$ and $\cPos$.

First we generalise the event order used for higher-dimensional automata~\cite{FJSZ_HDALang, FJSZ-Kleene} from a total order on antichains to an arbitrary event relation between non-parallel elements and equip c-posets with such event relations. Then we introduce labelled c-sets, c-posets and c-posets with event relations. Finally, we introduce a class of parametric full d-subcategories of labelled c-posets whose event relations are suitable for subsuming the base categories of various higher-dimensional automata.

Throughout this section, $\Sigma$ is a fixed alphabet.
For the following definitions, recall that we write $\parallel$ for the incomparability relation associated with a precedence of a c-poset.

\begin{definition}
\label{d:labelling-and-event-relation}
    Let $P$ be a c-poset.
    \begin{enumerate}
    \item
        A \emph{labelling} on $P$ is a function $\lambda:P\to \Sigma$.
    \item
        An \emph{event relation} on $P$ is a binary relation $\intord$ on $P$ such that $p\intord p' \implies p\parallel p'$.
    \end{enumerate}
\end{definition}

\begin{definition}
\label{d:preservation}
    Let $P$ and $Q$ be c-posets, $\lambda_P$ and $\lambda_Q$ labellings, and $\intord_P$ and $\intord_Q$ event relations on $P$ and $Q$, respectively. A c-poset map $f:Q\to P$
    \begin{enumerate}
    \item
        \emph{preserves labelling} if $\lambda_Q=\lambda_P\circ f$,
    \item 
        \emph{preserves event relation} if
        $q\parallel_Q q' \implies (q\intord_Qq'\iff f(q)\intord_P f(q'))$ for all $q,q'\in Q$.
    \end{enumerate}
\end{definition}

\begin{definition} {\ }
\label{d:enriched-posets}
\begin{enumerate}
\item 
    An \emph{lc-poset} (\emph{ec-poset}, \emph{lec-poset}) is a c-poset with a labelling (an event relation, a labelling and an event relation).
\item
    An \emph{lc-poset map} (\emph{ec-poset map}, \emph{lec-poset map}) is a c-poset map that preserves the labelling (event relation, both labelling and event relation).
    We write $\lcPos$ ($\ecPos$, $\lecPos$) for the d-category of lc-posets (ec-, lec-) and lc-poset maps (ec-, lec-), with the d-structure reflected from $\cPos$.
\item 
    An \emph{lc-set} (\emph{ec-set}, \emph{lec-set}) is a lc-poset (ec-poset, lec-poset) that has the empty precedence order. We write $\lcSet\hookrightarrow \lcPos$ for the full d-subcategory of lc-sets ($\ecSet\hookrightarrow\ecPos$ for ec-sets, $\lecSet\hookrightarrow\lecPos$ for lec-sets).
\end{enumerate}
\end{definition}

\begin{definition}
\label{d:v-poset}
	Let ${\cV}\rightarrow {\lecSet}$ be a sieve.
    A \emph{$\cV$-poset} is an lec-poset $P$ such that 
    $\lecPos(U,P)=\emptyset$ for all $U\in\Ob(\lecSet)\setminus \Ob(\cV)$.
    We write $\cV\Pos\hookrightarrow\lecPos$ for the full d-subcategory of $\cV$-posets.
\end{definition}

By definition, $\cV\Pos=\cV\setminus_{\lecPos}\lecSet$ is the residuation of $\cV$ and $\lecSet$ relative to $\lecPos$. It is thus a sieve by Lemma \ref{l:Q-domain-closed}.

Figure \ref{fig:poset-cats} shows the relationship between the d-categories introduced above. As for c-sets, lc-sets (ec-sets, lec-sets) are regarded as lc-posets (ec-posets, lec-posets) with empty precedence.

\begin{figure}
    \centering
    \begin{equation*}
\begin{tikzcd}[column sep=.5cm]
    \CSet
        \ar[rr, hook, shift left=3pt]
        \ar[dd,hook]
    &&
    \ecSet
        \ar[dd,hook]
        \ar[ll,shift left=3pt]
    &&&
\\
 &
    \lcSet
        \ar[rr, hook, crossing over, shift left=3pt] 
        \ar[ul]
    &&
    \lecSet
        \ar[dd, hook]
        \ar[ll,shift left=3pt, crossing over]
        \ar[ul]
    &&&
    \cV
        \ar[lll, hook]
        \ar[dd,hook]
\\
    \cPos
        \ar[rr,hook, shift left=3pt]
    &&
    \ecPos
        \ar[ll,shift left=3pt]
    &&&
\\
    &
    \lcPos
        \ar[rr,hook, shift left=3pt]
        \ar[from=uu, hook, crossing over]
        \ar[ul]
    &&
    \lecPos
        \ar[ll,shift left=3pt]
        \ar[ul]
    &&&
    \cV\Pos
        \ar[lll,hook]\\
\end{tikzcd}
\end{equation*}
    \caption{Categories of enriched sets and posets with event order and labels. Hooked arrows denote inclusions of full d-subcategories; the others\ are forgetful functors. Forgetful functors from c-posets to c-sets are omitted.}
    \label{fig:poset-cats}
\end{figure}
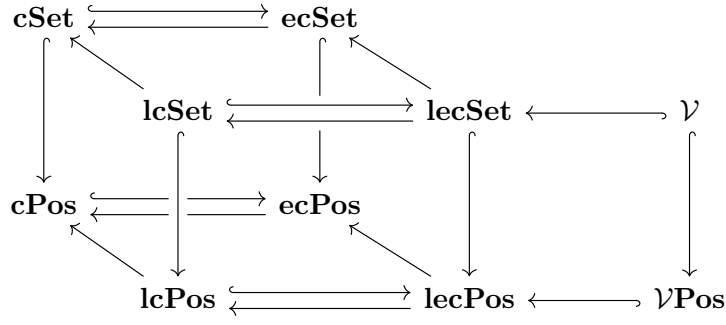


We now show that the categories $\cPos$, $\ePos$, $\lecPos$ and $\cV\Pos$ form a sequence in which successor categories are comma categories with respect to their predecessors. The category $\lcPos$ also fits into this picture, but we ignore it here.

First, we consider $\cPos$ and $\ecPos$.
Let $\mathbf{2}=\{0,1\}$.
We wish to encode event relations on a c-poset $P$ as $\mathbf{2}$-valued functions on pairs of $<$-incomparable elements.
To this end, we define a functor $\PE\colon \cPos\to\Set$ for any $P,Q\in\Ob(\cPos)$ and any c-poset map $f\colon Q\to P$ by
\begin{equation*}
    \PE(P)=\{(p,p')\in P\times P \mid p\parallel p' \}
    \qquad\text{and}\qquad
    \PE(f)(q,q')=(f(q),f(q')).    
\end{equation*}
This is indeed a functor since $f$ reflects precedence. Next, for any ec-poset $P$, let $\chi^P_{\intord}:\PE(P)\to\mathbf{2}$ the indicator map  sending each pair $(p,p')$ to $1$ if $p\intord p'$ and to $0$ otherwise. 

Using this data, we define the functor $\Phi\colon \ecPos\to \cPos\slice_\PE \mathbf{2}$ by 
\begin{equation*}
\Phi(P)=(\Forget^{\ePos}_{\cPos}(P),\chi^P_{\intord})\qquad\text{and}\qquad \Phi(f)=f 
\end{equation*}
for every ec-poset $P$ and ec-poset map $f$.
It is an isomorphism on objects since only parallel elements of $P$ can be related by $\intord$. For functoriality, suppose $f\colon Q\to P$ is a c-map between ec-posets. For every $(q,q')\in \PE(Q)$ we have
$\chi^P_{\intord}(\PE(f)(q,q'))=
\chi^P_{\intord}(f(q),f(q'))=1$ if and only if 
$f(q)\intord f(q')$. 
As $\chi^Q_{\intord}(q,q')=1$ if and only if $q\intord q'$ by definition. It therefore follows that
$\chi^P_{\intord}\circ \PE(f)=\chi^Q_{\intord}$ if and only if $f$ is an ec-poset map. Hence a c-poset map is an ec-poset map if and only if it is a morphism in $\cPos\slice_\PE 2$, and we have shown the following fact.

\begin{lemma}
\label{l:epos-is-slice}
The functor $\Phi\colon \ecPos\to \cPos\slice_\PE \mathbf{2}$  is an isomorphism of categories.
\end{lemma}

Encoding the labelling in a similar way is routine;  the next lemma is straightforward.

\begin{lemma}
\label{l:labelling-is-slice}
    The functor $\lecPos\to \ecPos\slice_{\Forget^{\ePos}_\Set}\Sigma$ given by $(P,\lambda)\mapsto (P,\lambda)$ is an isomorphism.
\end{lemma}

We summarise these relationships in the following diagram ($\cV\hookrightarrow\lecSet$ is a sieve as before).
\begin{equation}
\label{e:big-dcoco}
    \begin{tikzcd}[column sep=tiny]
        \CSet
        \ar[dd, hook]
        &&
        {\ecSet}
        \ar[ll]
        \ar[dd, hook]
        \ar[dl,"\cong"]
        &&
        \lecSet
        \ar[ll]
        \ar[dd, hook]
        \ar[dl,"\cong"]
        &&
        \cV
        \ar[ll,hook]
        \ar[dd, hook]
    \\
        &
        \CSet\slice_\PE \mathbf{2}
        \ar[ul]
                &&
        \ecSet\slice_{\Forget_\Set^\ePos}\Sigma
        \ar[ul]
        &&
        {\ \ }
    \\
        \cPos
        &&
        \ePos
        \ar[ll]
        \ar[dl,"\cong", "\Phi"']
        &&
        \lecPos
        \ar[ll]
        \ar[dl,"\cong"]
        &&
        \cV\Pos
        \ar[ll,hook]
    \\
        &
        \cPos\slice_\PE \mathbf{2}
        \ar[from=uu, hook, crossing over]
        \ar[ul]
        &&
        \ecPos\slice_{\Forget_\Set^\ePos}\Sigma
        \ar[from=uu, hook, crossing over]
        \ar[ul]
        &&
    \end{tikzcd}
\end{equation}
The arrows in the upper part of the diagram are simply the restrictions of the lower arrows to discrete c-posets.
All hooked arrows are inclusions of full d-subcategories, those labelled with $\cong$ are equivalences of d-categories by Lemmas \ref{l:epos-is-slice} and \ref{l:labelling-is-slice}, and the remaining are forgetful functors. Commutativity of the diagram is immediate from the definitions.


The final goal in this section is to show that all vertical functors in Diagram \eqref{e:big-dcoco} are d-cocomple\-tions. We know from Proposition~\ref{p:concset-cocompletion} that $\cPos$ is a d-cocompletion of $\CSet$. It remains to consider this relationship for the c-sets and c-posets introduced in this section: $\ecSet$  and $\ecPos$, $\lecSet$ and $\lecPos$ as well as $\cV$ and $\cV\Pos$. Applying Theorem \ref{t:SlicedCocompletion} requires checking that the d-functors between these pairs of d-categories preserve d-pushouts.

\begin{lemma}
\label{l:label-preserves-pushouts}
The functors $\PE\colon\cPos\to\CSet$ and  $\Forget^\ePos_\Set$ send d-pushouts to pushouts.
\end{lemma}
\begin{proof}
For $D$, consider a pushout $S= \colim\big(P\backleftarrow{f} R\forrightarrow{g} Q\big)$ in $\cPos$. Let $i\colon P\to S$ and $j\colon Q\to S$ denote the cocone c-poset maps. The construction of $<$ on $S$ in
 formula \eqref{e:concat-precedence} guarantees that $p<q$ whenever $p\in i(P)\setminus j(Q)$ and $q\in j(Q)\setminus i(P)$. Therefore,  
 \[
    \PE(S)\cong \PE(i)(\PE(P))\cup \PE(j)(\PE(Q)).
\]
Furthermore, $\PE(i)(\PE(P))\cap \PE(j)(\PE(Q))=\PE(if)(\PE(R))$ since $f$ and $g$ preserve and reflect $<$. 
Thus $\PE(S)\cong \colim(\PE(P)\xleftarrow{\PE(f)}\PE(R)\xrightarrow{\PE(g)}\PE(Q))$.
 
 For $\Forget^\ePos_\Set$, note that this functor is a composition
 \[
 \ePos\xrightarrow[\cong]{\Phi} \cPos\slice_\PE\mathbf{2}\xrightarrow{\ProjFun} \cPos\xrightarrow{\Forget}\cSet\xrightarrow{\Forget}\Set,
 \]
 where $\Phi$ is an isomorphism by Lemma \ref{l:epos-is-slice}, $\ProjFun$ and $\Forget^\cPos_\cSet$ preserve d-pushouts by Proposition \ref{p:slice-creates-pushouts} and Corollary \ref{t:cpos-has-proper-pushouts}, while $\Forget^\cSet_\Set$ sends d-pushouts to pushouts by Proposition \ref{p:cset-is-dcocomplete}.
\end{proof}

\begin{proposition}
\label{p:enriched-dcoco}
    The inclusions $\ecSet\hookrightarrow\ecPos$ and
    ${\lecSet}\hookrightarrow\lecPos$ are d-cocompletions, and  $\cV\hookrightarrow \cV\Pos$ is a d-cocompletion for every sieve $\cV\hookrightarrow\lecSet$.
\end{proposition}
\begin{proof}
    The inclusion $\CSet\hookrightarrow \cPos$ is a d-cocompletion by Proposition \ref{p:concset-cocompletion}. Since $\PE\colon \cSet\to \Set$ sends d-pushouts to pushouts (Lemma~\ref{l:label-preserves-pushouts}), we can  apply Theorem \ref{t:SlicedCocompletion} to show that the composition
    \[\ecSet\cong \CSet\slice_{\PE} \mathbf{2}
        \hookrightarrow
        \cPos\slice_{\PE} \mathbf{2}
        \cong\ePos
    \]
    is a d-cocompletion (the equivalences follow from Lemma \ref{l:epos-is-slice}).  
    Likewise, $\Forget^{\ecSet}_{\Set}$ sends d-pushouts to pushouts. Hence Theorem \ref{t:SlicedCocompletion} shows that the composition
    \[
        {\lecSet}\cong \ecSet\slice_{\Forget^\ecSet_\Set}\Sigma
        \hookrightarrow
        \ecPos\slice_{\Forget^\ecPos_\Set}\Sigma
        \cong
        \lecPos    
    \]
    is a d-cocompletion (these equivalences follow from Lemma~\ref{l:labelling-is-slice}). Finally, combining this result with Proposition \ref{p:subcat-cocompletion} shows that  $\cV\hookrightarrow\cV\Pos$ is a d-cocompletion.
\end{proof}

\section{Interval posets}
\label{s:interval-posets}

\def\MAC{\mathsf{M}}

We now prove our main result, which relates track objects on various categories of enriched c-sets with corresponding enrichments of interval posets.

First we recall the definition of interval posets. Let $\MAC(P)$ denote the set of all maximal antichains in the c-poset $P$. We define a partial order $\preceq$ on $\MAC(P)$ by
\[
A\preceq B \iff \forall a\in A, \forall b\in B.\ b\not<a.
\]
\begin{definition}
\label{d:interval-order}
     An \emph{interval c-poset} is a c-poset $P$ in which $\MAC(P)$ is totally ordered by $\preceq$.
\end{definition}
For this and equivalent definitions of interval orders, as well as more on this topic, see \cite{Fishburn-Interval_Orders}.

Interval c-posets form the objects of the full d-subcategory $\cIPos$ of $\cPos$.  
Likewise, we write $\ecIPos\hookrightarrow \ecPos$, $\lecIPos\hookrightarrow\lecPos$ and $\cV\IPos\hookrightarrow\cV\Pos$ for the full d-subcategories of interval ec-posets, interval lec-posets and interval $\cV$-posets, respectively.

\begin{lemma}
\label{l:interval-concat}
Let $P$ and $Q$ be interval c-posets, let $U$ be a c-set. Then $\colim\big(P\backleftarrow{f} U \forrightarrow{g} Q\big)$ is an interval c-poset. The bicategory $\dCosp(\cIPos,\cSet)$ is thus well-defined.
\end{lemma}

\begin{proof}
Let $S=\colim\big(P\backleftarrow{f} U \forrightarrow{g} Q\big)$,
let $P\forrightarrow{i}S$ and $Q\backrightarrow{j}S$ be structural c-maps as in diagram \eqref{e:cposet-concatenation}; recall that both are precedence embeddings. By construction of d-pushouts, in particular the formula \eqref{e:concat-precedence}, if $p\in S\setminus j(Q)$ and $q\in S\setminus i(P)$, then $p<q$. Thus every maximal antichain in $S$ is either contained in $i(P)$ or in $j(P)$.
For maximal antichains $A,B\subseteq S$,  we show that $A\preceq B$ or $B\preceq A$. 
If $A,B\subseteq i(P)$, then $i^{-1}(A)$ and $i^{-1}(B)$ are antichains in $P$ and then $i^{-1}(A)\preceq i^{-1}(B)$ or $i^{-1}(B)\preceq i^{-1}(A)$ because $\preceq$ is assumed to be total on $\MAC(P)$. Thus $A\preceq B$ or $B\preceq A$. If $A,B\subseteq j(Q)$, a similar argument applies. 
Otherwise, for $A\not\subseteq j(Q)$ and $B\not\subseteq i(P)$, suppose $a\in A$ and $b\in B$. Clearly, $a\in i(P)$ and $b\in j(Q)$. If $b<a$, then $a,b\in i(P)\cap j(Q)$ since $i(P)$ and $j(Q)$ are a down-closed and an up-closed subset of $S$, respectively. But $i(P)\cap j(Q)=if(U)$ is discrete -- a contradiction. Therefore, $b\not<a$ and $A\preceq B$.
\end{proof}

\begin{lemma}
\label{l:interval-poset-decomposition}
    Let $P$ be an interval c-poset that is not a c-set. Then there exists a d-pushout of c-posets
\[
    \begin{tikzcd}
        U\ar[r,"f","-" swap]\ar[d,"g" swap,"+"] & Q\ar[d,"i","+" swap] \\
        R\ar[r,"j" swap,"-"] & P 
    \end{tikzcd}
\]
    such that $U$ is a c-set, $Q$ and $R$ are interval c-posets and $P_\mu\subseteq i(Q)\neq P\neq j(R)\supseteq P_\nu$, where $P_\mu$ and $P_\nu$ be the sets of minimal and maximal elements of $P$, respectively.
\end{lemma}
\begin{proof}
Consider the sub-posets 
$Q=P_\mu$,  $R=\bigcup \{A\in \MAC(P)\mid A\neq P_\mu\}$ and $U=Q\cap R$ of $P$. Let $f,g,i,j$ be c-poset maps given by suitable inclusions such that $g,i$ are formorphisms and $f,j$ backmorphisms. Then $Q$ and $U$ are c-sets and $Q\neq P$. Furthermore, $P_\nu\subseteq R$ since $P_\nu\neq Q$. Assume that $R=P$. Then there are $q,q'\in Q$ and $A\neq B\in \MAC(P)\setminus\{Q\}$ such that $q\in A\setminus B$ and $q'\in B\setminus A$. Furthermore, there must exist a $p\in A$ such that $q'<p$ and a $p'\in B$ such that $q<p'$. The antichains $A$ and $B$ are therefore not comparable by $\preceq$, and  $R\neq P$ by contradiction. 

For each $q\in Q\setminus R$ and $r\in R\setminus Q$ we have $q<r$. Otherwise, the antichain $\{q,r\}$ could be extended to a maximal antichain $U\neq Q$ -- a contradiction. The relations \eqref{e:concat-precedence} defining the precedence of the d-pushout are therefore satisfied and $P\cong \colim(Q\backleftarrow{\supset} U \forrightarrow{\subset} R)$.
\end{proof}

\begin{proposition}
\label{p:interval-cospan-poset}
The bicategories  $\dICosp(\cPos,\cSet)$ and $\dCosp(\cIPos,\cSet)$ are equal sub-bicategories of $\dCosp(\cPos,\cSet)$.
    Equivalently, a d-cospan  $\big(U\forrightarrow{f}P \backleftarrow{g} V\big)$ in $\dCosp(\cPos,\cSet)$ is interval if and only if $P$ is an interval c-poset. 
\end{proposition}

\begin{proof}
    Since $\dCosp(\cIPos,\cSet)$ contains all elementary d-cospans and is locally isomorphism-closed, we have $\dICosp(\cPos, \cSet)\subseteq \dCosp(\cIPos, \cSet)$.

    Let $\mathscr{P}$ be the class of objects $P$ of $\cIPos$ such
    that every d-cospan $\big(U\forrightarrow{} P\backleftarrow{}
    V\big)$ for all $U,V\in\Ob(\cSet)$ is interval and hence in $\dICosp(\cPos,\cSet)$. Clearly, all c-sets belong to $\mathscr{P}$. We show that $\Ob(\cIPos)\subseteq \mathscr{P}$ by induction on the number of elements of $P\in\Ob(\cIPos)$. For $P=\emptyset$ this is already shown. For the inductive step, there two cases. If $P$ is a c-set, then the d-cospan is elementary and hence interval. Otherwise, $P$ can be decomposed as 
    \[
        P=\colim\big(Q\backleftarrow{f} U \forrightarrow{g} R\big)
    \]
    as in Lemma \ref{l:interval-poset-decomposition}. Now assume that $(V\forrightarrow{h} P \backleftarrow{k} W)$ is a d-cospan, as in the diagram
    
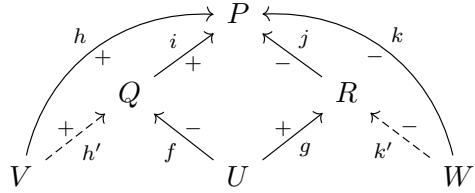
\begin{equation*}
    \begin{tikzcd}[row sep=0.5cm, column sep=0.8cm]
        && P&&\\
        &Q\ar[ur, "i", "+"'{inner sep=0pt}] && R\ar[ul, "j"',"-"{inner sep=0pt}] &&\\
    V\ar[uurr, bend left=40, "h","+"'{inner sep=0pt}]\ar[ur, swap, dashed, "h'","+"'{inner sep=0pt}] && U\ar[ul, "f","-"'{inner sep=0pt}]\ar[ur, swap, "g","+"'{inner sep=0pt}] && W\ar[uull, bend right=40, "k"',"-"{inner sep=0pt}]\ar[ul, dashed, "k'","-"'{inner sep=0pt}]
    \end{tikzcd}
    \end{equation*}
    Since $i(Q)$ contains all minimal elements of $P$, there is a factorisation $h=ih'$ in $\Set$ that extends uniquely to a factorisation in $\cSet$ (we make $h'$ into a formorphism). Similarly, there is a factorisation $k=jk'$, where $k'$ is a backmorphism.
    
    By the inductive hypothesis, $Q,R\in\mathscr{P}$. Hence the lower d-cospans in the diagram above are interval, which implies that the big d-cospan is interval, too. As this holds for arbitrary choices of $V$, $W$, $f$ and $g$, it follows that $P\in\mathscr{P}$. Thus $\Ob(\cIPos)\subseteq \mathscr{P}$ and finally $\dCosp(\cIPos, \cSet)\subseteq \dICosp(\cPos,\cSet)$.
\end{proof}

Our main result follows. 

\begin{theorem}
\label{t:cospans-are-track-objects}
    Let $\cV\hookrightarrow\lecSet$ be a sieve. Then there are the following sequences of natural biequivalences:
\begin{enumerate}
    \item
        $\dCosp(\cIPos,\cSet)
        \cong
        \dICosp(\cPos,\cSet)
        \xrightarrow{\cong}
        \tro\cSet,
        $
    \item
        $\dCosp(\ecIPos,\ecSet)
        \cong
        \dICosp(\ecPos,\ecSet)
        \xrightarrow{\cong}
        \tro\ecSet,
        $
    \item
        $\dCosp(\lecIPos,\lecSet)
        \cong
        \dICosp(\lecPos,\lecSet)
        \xrightarrow{\cong}
        \tro\lecSet,
        $
    \item
        $\dCosp(\cV\IPos,\cV)
        \cong
        \dICosp(\cV\Pos,\cV)
        \xrightarrow{\cong}
        \tro\cV.
        $
\end{enumerate}
\end{theorem}
\begin{proof}
    The right-hand equivalences follow from Theorem \ref{t:dcospans-are-track-objects}, since the suitable inclusions are d-cocom\-pletions by Propositions \ref{p:concset-cocompletion} and \ref{p:enriched-dcoco}.
    It remains to prove the left-hand equivalences.

    For (1), the equation is Proposition \ref{p:interval-cospan-poset}.
    For (2), let $\Forget:\ecPos\cong \cPos\slice_D\mathbf{2}\xrightarrow{\ProjFun} \cPos$ be the forgetful functor. Consider the diagram
    \[
    \begin{tikzcd}
        \dICosp(\ecPos,\ecSet)
            \ar[r,hook]
            \ar[d,"\Forget_*"]
    &
        \dCosp(\ecPos,\ecSet)
            \ar[d,"\Forget_*"]
    &
        \dCosp(\ecIPos,\ecSet)
            \ar[l,hook]
            \ar[d,"\Forget_*"]
    \\
        \dICosp(\cPos,\cSet)
            \ar[r,hook]
    &
        \dCosp(\cPos,\cSet)
    &
        \dCosp(\cIPos,\cSet)
            \ar[l,hook]
    \end{tikzcd}
    \]
    The following equations between sub-bicategories of $\dCosp(\ecPos,\ecSet)$ prove the claim:
    \begin{align*}
     \dICosp(\ecPos,\ecSet)
     &=
     (\Forget_*)^{-1}(\dICosp(\cPos,\cSet))
     \tag{by Proposition \ref{p:interval-slice}}
     \\
     &=
     (\Forget_*)^{-1}(\dCosp(\cIPos,\cSet)) \tag{by (1)}\\
     &=
     \dCosp(\ecIPos,\ecSet).
     \tag{by definition}
    \end{align*}
    Repeating this argument with $\lecPos\cong \ecPos\slice_{\Forget^\ecPos_\Set}\Sigma$ and using (2) shows (3). Finally, since $\cV\Pos\cong \cV\setminus_{\lecPos}\lecSet$,   we obtain the following sequence of equalities of sub-bicategories of $\dCosp(\lecPos,\lecSet)$:
    \begin{align*}
     \dICosp(\cV\Pos,\cV)
     &=
     \dICosp(\lecPos,\lecSet)\cap \dCosp(\cV\Pos,\cV)
     \tag{by Lemma \ref{l:interval-sieve}}
     \\
     &=
     \dCosp(\lecIPos,\cSet)\cap \dCosp(\cV\Pos,\cV) \tag{by (3)}\\
     &=
     \dCosp(\cV\IPos,\cV),
     \tag{by definition}
    \end{align*}    
    which proves (4).
\end{proof}

Theorem~\ref{t:cospans-are-track-objects} provides an explicit description of bicategories of track objects over d-subcategories of the category of lec-posets $\lecSet$. Its applications in concurrency theory via higher-dimensional automata are explored in the next section.

\section{Variants of higher-dimensional automata}
\label{s:HDA}

In this final section we list several base d-categories for presheaf automata that provide variants of higher-dimensional automata. We instantiate Theorem \ref{t:cospans-are-track-objects} to describe their track objects and hence their languages.  Some of these variants appear in the literature~\cite{FJSZ_HDALang, FJSZ-Kleene, Gaucher-Symmetric.Precubical.Sets}; others are new.

Throughout this section, we fix again an alphabet $\Sigma$. 
We adjust terminology relative to earlier sections (for example, we avoid using cospans and c-maps) to make this section more accessible and self-contained.

A \emph{higher-dimensional automaton with event relations} (\emph{erHDA}) is a presheaf automaton over the category $\lecSet$ of labelled ec-sets.
This kind of higher-dimensional automaton matters because it subsumes all other variants of HDA in this section; the event relation, in particular, is weaker than usual. Whether erHDA are useful in applications remains open. 

Recall that the inclusion $\lecSet\hookrightarrow\lecPos$ is a d-cocompletion (Proposition \ref{p:enriched-dcoco}) and that the bicategory $\tro\lecSet$ of track objects over $\lecSet$ is biequivalent to $\dCosp(\lecIPos, \lecSet)$ (Theorem \ref{t:cospans-are-track-objects}).
We now give a straightforward description of this bicategory, collecting definitions. Recall that 
\begin{itemize}
\item 
    A lec-set $(U,{\intord_U},\lambda_U)$ is formed by  a finite set $U$,  an irreflexive relation $\intord$ called \emph{event relation} and  a labelling $\lambda_U:U\to \Sigma$.    
\item 
    A lec-poset $(P,{<_P},{\intord_P},\lambda_P)$ is a lec-set $(P,{\intord_P},\lambda_P)$ equipped with is a strict partial order $<_P$ called \emph{precedence} that is disjoint with the event relation: if $p<_Pp'$ then neither $p\intord_P$ nor $p'\intord_Pp'$.
\end{itemize}

\begin{definition}
    A \emph{lec-poset with interfaces} (\emph{lec-iposet}) is a structure $\big(U\xrightarrow{s}P\xleftarrow{t} V\big)$, where $P$ is a lec-poset, $U$ and $V$ are lec-sets, and $s$ and $t$ are injective functions, called  \emph{source} and  \emph{target interface}, that
\begin{enumerate}
    \item preserve the labelling ($\lambda_U=\lambda_P\circ s$ and $\lambda_V=\lambda_P\circ t$),
    \item preserve and reflect the event relation: for all $u,u\in U$ and $ v,v'\in V$,
    \[
        u\intord_Uu'\iff s(u)\intord_Ps(u')
        \quad \text{and} \quad
        v\intord_Vv'\iff t(v)\intord_Pt(v'),
    \]
    \item all elements of $s(U)$ are minimal in $P$ and all elements of $t(V)$ are maximal in $P$ with respect to $<_P$.
\end{enumerate}
The lec-iposet $\big(U\xrightarrow{s}P\xleftarrow{t} V\big)$ is interval if $(P,{<_P})$ is an interval order (see Definition \ref{d:interval-order}).
\end{definition}

\begin{definition}
    A \emph{subsumption} between lec-iposets $(U\xrightarrow{s_Q}Q\xleftarrow{t_Q}V)$ and $(U\xrightarrow{s_P}P\xleftarrow{t_P}V)$ is a bijection $f:Q\to P$ that
    \begin{enumerate}
        \item preserves the interfaces ($f\circ s_Q=s_P$ and $f\circ t_Q=t_P$),
        \item reflects precedence ($\forall q,q'\in Q.\ f(q)<_P f(q') \implies q<_Qq'$),
        \item preserves and reflects essential event relation: for all $q,q'\in Q$,
        \[
        (q<_Qq')\;\lor\; (q'<_Qq)\;\lor\; (q\intord_Qq'\iff f(q)\intord_Pf(q')).
        \]
    \end{enumerate}
\end{definition}

\begin{definition}
The \emph{concatenation} of lec-iposets
\[
(U\forrightarrow{s_P}P\backleftarrow{t_P}V)
*
(V\forrightarrow{s_Q}Q\backleftarrow{t_Q}W)
=
(U\forrightarrow{s}P*Q\backleftarrow{t}W)
\]
is defined as
\begin{enumerate}
    \item $P*Q = \colim(P\xleftarrow{t_P}V\xrightarrow{s_Q}Q)=(P\sqcup Q)/(t_P(v)\sim s_Q(v))$ (where $i:P\to P*Q$ and $j:Q\to P*Q$ are the canonical inclusions),
    \item $r<_{P*Q}r' \iff (r \mathrel{i(<_P)} r' \;\lor\; r \mathrel{j(<_Q)} r' \;\lor\; (r\notin j(Q)\land r'\notin i(P)))$,
    \item $r\intord_{P*Q}r' \iff (r \mathrel{i(\intord_P)} r' \;\lor\; r \mathrel{j(\intord_Q)} r')$,
    \item $\lambda_{P*Q}\circ i=\lambda_P$ and $\lambda_{P*Q}\circ j=\lambda_Q$,
    \item $s=i\circ s_P$ and $t=j\circ t_Q$.
\end{enumerate}
\end{definition}

Interval lec-iposets are equivalent to d-cospans in $\dCosp(\lecIPos, \lecSet)$ by Definition \ref{d:enriched-posets}. The functions $s$ and $t$ correspond uniquely to  for- and backmorphisms by Lemma \ref{l:cposet-map-properties}.4.

Subsumptions are precisely morphisms in the hom categories $\dCosp(\lecPos,\lecSet)_U^V$.
A lec-poset map $f$, as in \eqref{e:dcospan-map}, must be both a for- and a backmorphism by Lemma \ref{l:cset-properties}.\ref{i:pure}; hence it is uniquely determined by a bijection.

It is routine to verify that concatenation of lec-iposets agrees with horizontal composition in $\dCosp(\lecIPos, \lecSet)$ using the construction in Theorem \ref{t:cpos-has-proper-pushouts} for (1) and (2), and the fact that the canonical inclusions of a pushout in $\lecPos$ are lec-poset maps. We obtain the following result.

\def\IP{\mathsf{IP}}
\begin{proposition}
    The bicategory $\tro{\lecSet}$ is biequivalent to the bicategory $\mathsf{\IP}_{lec}$ that has lec-sets as objects. Its hom category $(\mathsf{\IP}_{lec})_U^V$ has interval lec-iposets $(U\to P\gets V)$ as objects and subsumptions as morphisms. Vertical composition is composition of subsumptions, and horizontal composition is given by concatenation.
\end{proposition}

Since every subsumption is a precedence-reflecting bijection,  every endo-subsumption is an isomorphism. This leads to the following description of the category of words (see Section \ref{s:preliminaries}).
\begin{proposition}
    The category $\mathscr{W}(\lecSet)$ has lec-sets as objects and isomorphisms classes of lec-iposets in $(\mathsf{\IP}_{lec})_U^V$ as morphisms from $U$ to $V$. Moreover, $(U\to P \gets V)\subsu (U\to Q \gets V)$ if there is a subsumption from $(U\to P \gets V)$ to $(U\to Q \gets V)$.
\end{proposition}

With erHDA set up, we proceed as in Sections~\ref{s:enriched-posets} and \ref{s:interval-posets} and use enriched $\cV$-parametrised posets to generate variants. So let $\cV\hookrightarrow\lecPos$ be a sieve, which can  be regarded as a class of lec-sets closed with respect to taking subsets (with the inherited labelling and event relation).

    \begin{definition}
        A \emph{$\cV$-iposet} is a lec-iposet $(S\to P\gets T)$ such that every antichain in $P$ is in $\cV$.
    \end{definition}
    This is equivalent to $P$ being a $\cV$-poset: an lec-poset map $U\to P$ from a lec-set $U$ exists if and only if $P$ has an antichain isomorphic to $U$.

\begin{proposition}
\label{p:language-of-V}
    Let $\cV\hookrightarrow\lecSet$ be a sieve.
    \begin{enumerate}
    \item 
        The bicategory $\tro\cV$ is biequivalent to the sub-bicategory $\IP_{\cV}$ of $\IP_{\mathbf{lec}}$ with objects $\Ob(\cV)$ and where the hom categories $(\IP_{\cV})_U^V\subseteq (\IP_{\mathbf{lec}})_U^V$ are full subcategories with$\cV$-iposets as objects.
    \item 
        The category of words on $\cV$ is the full subcategory $\mathscr{W}(\cV)$ of $\mathscr{W}(\lecSet)$ with objects $\cV$ and isomorphism classes of $\cV$-iposets as morphisms.
    \end{enumerate}
\end{proposition}  
\begin{proof}
    Item (1) is a reformulation of Theorem \ref{t:cospans-are-track-objects}, (2) is an immediate consequence.
\end{proof}

\begin{paragraph}{Higher-dimensional automata (HDA).}

A \emph{conclist}~\cite{FJSZ-Kleene} is an lec-set $U$ that is totally ordered by the event relation $\intord$. Let $\square\hookrightarrow\lecSet$ be a sieve that has conclists as objects. 

Proposition \ref{p:language-of-V} tells us that the bicategory  $\tro{\square}$ is biequivalent to the sub-bicategory $\IP_\square\subseteq \IP_{\mathbf{lec}}$ of interval $\square$-iposets,
where $\square$-iposets are those lec-iposets $(S\to P\gets T)$ for which the event relation on $P$ is a total order on antichains of the precedence. Total event orders allow at most one isomorphism between any two objects of $\square$. The source and target interface maps are thus uniquely determined by their images. Both interfaces of $P$ can therefore be regarded as a minimal and a maximal subset of $P$, as in~\cite{FJSZ_HDALang}.

Presheaf automata over $\square$ are precisely the HDA of \cite{FJSZ_HDALang,FJSZ-Kleene}.
\end{paragraph}

\begin{paragraph}{Symmetric higher-dimensional automata -- sHDA. }
A \emph{concset} is an lc-set regarded, for the sake of uniformity, as an lec-set with empty event relation. 

Let $\Xi\hookrightarrow\lecSet$ be the sieve of concsets. A \emph{symmetric higher-dimensional automaton} (\emph{sHDA}) is a presheaf automaton over $\Xi$. 
The bicategory $\tro{\Xi}$  is biequivalent to $\IP(\Xi)$, where $\Xi$-iposets are lec-iposets with empty event relation. In the absence of an event relation, interfaces of $\Xi$-iposets are not determined by their images.

Symmetric HDA and their variants were studied in \cite{Gaucher-Symmetric.Precubical.Sets,Goubault-Mimram_Relationship,Kahl-SymmetricHDA,Zouari-Forgetting}, among others; Grandis and Mauri's restricted cubical site $\mathbb{I}$~\cite{Grandis-Mauri_Cubical}, as a strict symmetric monoidal category, is strongly related.
\end{paragraph}

The event order in HDA is not merely about the difference between concset and conclist data structures, with the event order being the natural order of list elements. It also prescribes how the posets associated to HDA through their track objects compose along interfaces. Although event orders sometimes lead to simpler proofs, for instance in the Kleene theorem for HDA~\cite{FJSZ-Kleene}, their presence in concurrent systems seems debatable. Symmetric HDA involve less structure and they offer greater flexibility with compositions, yet track objects are no longer represented by posets with subsets as interfaces, but rather by cospans. This increases the combinatorial complexity in calculations; symmetry groups are now needed to define compositions. Next we propose a new kind of HDA in which the respective merits of the two previous models are combined.

\def\mxConcset{\mathsf{M}}

\begin{paragraph}{Mixed higher-dimensional automata (mHDA).}
 A \emph{mixed concset} is an lec-set $U$ such that
    \begin{itemize}
    \item
        $\forall u,v\in U.\ u\intord v \implies \lambda(u)=\lambda(v)$,
    \item
         $\intord$ is a strict total order on $\lambda^{-1}(a)\subseteq U$ for each $a\in\Sigma$.
    \end{itemize}
    Let $\mxConcset\hookrightarrow \lecSet$ be the sieve of mixed concsets.
    A \emph{mixed HDA} is a presheaf automaton over $\mxConcset$.
    
    The event order is now only used to distinguish different elements with the same label, which suffices to ensure that isomorphisms in $\mxConcset$ are unique. 
    Track objects in $\tro{\mxConcset}\cong \IP_\mxConcset$ are those interval lec-iposets $(S\xrightarrow{s} P\xleftarrow{t} T)$ in which antichains are mixed concsets. 
    Since interfaces are determined by their images, they can be regarded just as subsets $S$ and $T$ of the minimal and maximal elements of $P$, respectively.
\end{paragraph}

\def\ordConcset{\mathsf{\Theta}}
\begin{paragraph}{Ordered HDA (oHDA)}
Fix a total order $\prec$ on the alphabet $\Sigma$. A conclist $U$ is \emph{ordered} if
\[
    \forall{u,v\in U}.\ \lambda(u)\prec \lambda(v) \implies u\intord v.
\]
Let $\ordConcset\hookrightarrow\lecSet$ be the sieve of ordered conclists. An \emph{ordered HDA} is a presheaf automaton over $\ordConcset$.
Track objects in $\tro{\ordConcset}\cong \IP_\ordConcset\subseteq \IP_{\square}$ are those interval $\square$-iposets $(S\xrightarrow{s} P\xleftarrow{t} T)$ in which the event order agrees with the chosen order on $\Sigma$. 

The d-categories $\ordConcset$ and $\mxConcset$ are isomorphic. Removing event relation between elements with different labels in an ordered conclist yields a mixed concset. The inverse functor ``fills'' the missing event relations using the ordering in $\Sigma$.
Ordered and mixed HDA therefore form equivalent categories;  their bicategories of track objects are equivalent, too. An advantage of ordered HDA over mixed HDA is that they form a subcategory of HDA; a disadvantage is that they depend on the choice of an order on the alphabet.
\end{paragraph}

\def\nscConcset{\mathsf{N}}
\begin{paragraph}{Non-self-concurrent HDA (nHDA).}
    A concset $U$ is \emph{non-self-concurrent} if its labelling is injective. Let $\nscConcset\hookrightarrow \lecSet$ be the sieve of non-self-concurrent concsets, and presheaf automata over $\nscConcset$ are referred to as \emph{non-self-concurrent HDA}.

    Track objects in $\tro{\nscConcset}\cong \IP_\nscConcset\subseteq \IP_{\Xi}$ are those interval $\Xi$-iposets $(S\xrightarrow{s} P\xleftarrow{t} T)$ in which parallel events always have different labels.

    Clearly, $\nscConcset\hookrightarrow \mxConcset$, and  for every order on $\Sigma$ we have inclusions $\nscConcset\hookrightarrow \ordConcset\hookrightarrow\square$. Thus non-self-concurrent automata are a special case of all variants of HDA considered here.

    While nHDA are less expressive that other variants, most examples we know of do not require self-concurrency and hence can be modelled by nHDA.
\end{paragraph}

\begin{figure}
\centering
\begin{tikzpicture}[scale=1.5]
    \node(02) at (0,2) {$a$};
    \node(01) at (0,1) {$a$};
    \node(00) at (0,0) {$b$};
    \node(12) at (1,2) {$a$};
    \node(11) at (1,1) {$b$};
    \node(21) at (2,1) {$a$};
    \draw[->] (02)--(12);
    \draw[->] (01)--(11);
    \draw[->] (02)--(11);
    \draw[->] (11)--(21);
    \draw[->] (00)--(21);
    \draw[dashed,->] (02)--(01);
    \draw[dashed,->] (01)--(00);
    \draw[dashed,->] (12)--(01);
    \draw[dashed,->] (12)--(11);
    \draw[dashed,->] (12)--(21);
    \draw[dashed,->] (12)--(00);
    \draw[dashed,->] (11)--(00);
    \draw[dashed,->] (02) to[bend right=30] (00);
    \node(s0) at (-1,0) {$b$};
    \node(s1) at (-1,1) {$a$};
    \node(t1) at (3,1) {$a$};
    \draw[dashed,->] (s1)--(s0);
    \draw[->,dotted] (s1)--(01);
    \draw[->,dotted] (s0)--(00);
    \draw[->,dotted] (t1)--(21);
\end{tikzpicture}
\hspace{3em}
\begin{tikzpicture}[scale=1.5]
    \node(02) at (0,2) {$a$};
    \node(01) at (0,1) {$a$};
    \node(00) at (0,0) {$b$};
    \node(12) at (1,2) {$a$};
    \node(11) at (1,1) {$b$};
    \node(21) at (2,1) {$a$};
    \draw[->] (02)--(12);
    \draw[->] (01)--(11);
    \draw[->] (02)--(11);
    \draw[->] (11)--(21);
    \draw[->] (00)--(21);
    \draw[dashed,->] (02)--(01);
    \draw[dashed,->] (12)--(01);
    \draw[dashed,->] (12)--(21);
    \draw[dashed,->] (11)--(00);
%
    \node(s0) at (-1,0) {$b$};
    \node(s1) at (-1,1) {$a$};
    \node(t1) at (3,1) {$a$};
%
    \draw[->,dotted] (s1)--(01);
    \draw[->,dotted] (s0)--(00);
    \draw[->,dotted] (t1)--(21);
\end{tikzpicture}                
\caption{An example of a $\ordConcset$-iposet with order $a\prec b$ (left) and the corresponding $\mxConcset$-iposet (right). Solid arrows show precedence; dashed arrows show event relation.}
\label{fig:ipomsets}
\end{figure}
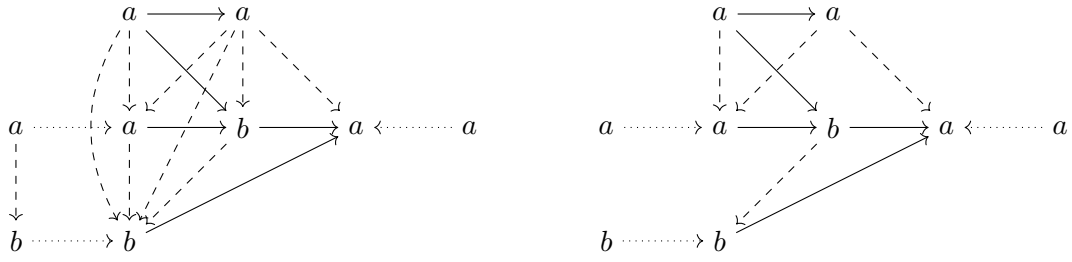

\begin{paragraph}{Comparison}
The base d-categories of the variants of higher-dimensional automata introduced above form the diagram of functors shown on the left below. The diagram on the right captures the relationships between the corresponding higher-dimensional automata.
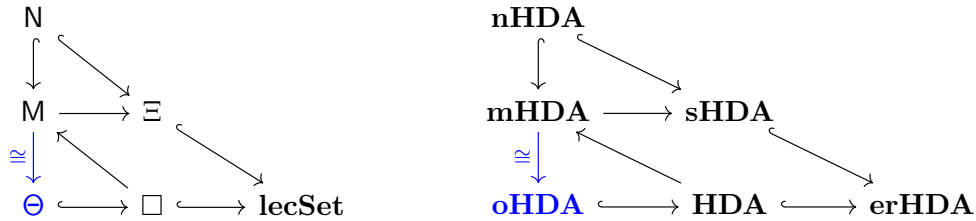
\begin{equation*}
\begin{tikzcd}[row sep=normal]
    \nscConcset
        \ar[d,hook]
        \ar[dr,hook]
\\
    \mxConcset
        \ar[r]
        \ar[d,"\cong"',color=blue]
    &
    \Xi
        \ar[dr,hook]
\\
    {\color{blue}\ordConcset}
        \ar[r,hook]
    &
    \square
        \ar[r,hook]
        \ar[ul]
    &
    \lecSet
\end{tikzcd}
\qquad\qquad
\begin{tikzcd}[row sep=normal]
    \mathbf{nHDA}
        \ar[d,hook]
        \ar[dr,hook]
\\
    \mathbf{mHDA}
        \ar[r]
        \ar[d,"\cong"',color=blue]
    &
    \mathbf{sHDA}
        \ar[dr,hook]
\\
    {\color{blue}\mathbf{oHDA}}
        \ar[r,hook]
    &
    \mathbf{HDA}
        \ar[r,hook]
        \ar[ul]
    &
    \mathbf{erHDA}
\end{tikzcd}
\end{equation*}
Objects that depend on the choice of order on $\Sigma$ are marked in blue. The diagrams shown do not commute (except for the top triangles).
\end{paragraph}

\printbibliography

\end{document}